\documentclass[a4paper,11pt]{article}
\pdfoutput=1 
\usepackage[T1]{fontenc} 

\makeatletter
\gdef\@fpheader{ }
\gdef\@journal{ }
\makeatother
\RequirePackage{amsmath}
\RequirePackage{amssymb}
\RequirePackage{epsfig}
\RequirePackage{graphicx}
\RequirePackage[numbers,sort&compress]{natbib}
\RequirePackage{color}
\RequirePackage[colorlinks=true
,urlcolor=blue
,anchorcolor=blue
,citecolor=blue
,filecolor=blue
,linkcolor=blue
,menucolor=blue
,pagecolor=blue
,linktocpage=true
,pdfproducer=medialab
,pdfa=true
]{hyperref}

\newif\ifnotoc\notocfalse
\newif\ifemailadd\emailaddfalse
\newif\iftoccontinuous\toccontinuousfalse
\makeatletter
\def\@subheader{\@empty}
\def\@keywords{\@empty}
\def\@abstract{\@empty}
\def\@xtum{\@empty}
\def\@dedicated{\@empty}
\def\@arxivnumber{\@empty}
\def\@collaboration{\@empty}
\def\@collaborationImg{\@empty}
\def\@proceeding{\@empty}
\def\@preprint{\@empty}

\newcommand{\subheader}[1]{\gdef\@subheader{#1}}
\newcommand{\keywords}[1]{\if!\@keywords!\gdef\@keywords{#1}\else%
\PackageWarningNoLine{\jname}{Keywords already defined.\MessageBreak Ignoring last definition.}\fi}
\renewcommand{\abstract}[1]{\gdef\@abstract{#1}}
\newcommand{\dedicated}[1]{\gdef\@dedicated{#1}}
\newcommand{\arxivnumber}[1]{\gdef\@arxivnumber{#1}}
\newcommand{\proceeding}[1]{\gdef\@proceeding{#1}}
\newcommand{\xtumfont}[1]{\textsc{#1}}
\newcommand{\correctionref}[3]{\gdef\@xtum{\xtumfont{#1} \href{#2}{#3}}}
\newcommand\jname{JHEP}
\newcommand\acknowledgments{\section*{Acknowledgments}}

\newcommand\preprint[1]{\gdef\@preprint{\hfill #1}}

\newtheorem{theorem}{Theorem}

\newtheorem{lemma}[theorem]{Lemma}

\newenvironment{proof}[1][Proof]{\noindent\textbf{#1.} }{\ \rule{0.5em}{0.5em}}

\makeatother

\newcommand\note[2][]{%
\if!#1!%
\stepcounter{footnote}\footnotetext{#2}%
\else%
{\renewcommand\thefootnote{#1}%
\footnotetext{#2}}%
\fi}

\makeatletter
\newtoks\auth@toks
\renewcommand{\author}[2][]{%
  \if!#1!%
    \auth@toks=\expandafter{\the\auth@toks#2\ }%
  \else
    \auth@toks=\expandafter{\the\auth@toks#2$^{#1}$\ }%
  \fi
}
\makeatother
\makeatletter
\newtoks\affil@toks\newif\ifaffil\affilfalse
\newcommand{\affiliation}[2][]{%
\affiltrue
  \if!#1!%
    \affil@toks=\expandafter{\the\affil@toks{\item[]#2}}%
  \else
    \affil@toks=\expandafter{\the\affil@toks{\item[$^{#1}$]#2}}%
  \fi
}
\makeatother
\makeatletter
\newtoks\email@toks\newcounter{email@counter}%
\newcommand{\emailAdd}[1]{%
\emailaddtrue%
\ifnum\theemail@counter>0\email@toks=\expandafter{\the\email@toks, \@email{#1}}%
\else\email@toks=\expandafter{\the\email@toks\@email{#1}}%
\fi\stepcounter{email@counter}}
\newcommand{\@email}[1]{\href{mailto:#1}{\tt #1}}
\makeatother

\makeatletter
\newcommand*\collaboration[1]{\gdef\@collaboration{#1}}
\newcommand*\collaborationImg[2][]{\gdef\@collaborationImg{#2}}
\makeatletter
\newcommand\afterLogoSpace{\smallskip}
\newcommand\afterSubheaderSpace{\vskip3pt plus 2pt minus 1pt}
\newcommand\afterProceedingsSpace{\vskip21pt plus0.4fil minus15pt}
\newcommand\afterTitleSpace{\vskip23pt plus0.06fil minus13pt}
\newcommand\afterRuleSpace{\vskip23pt plus0.06fil minus13pt}
\newcommand\afterCollaborationSpace{\vskip3pt plus 2pt minus 1pt}
\newcommand\afterCollaborationImgSpace{\vskip3pt plus 2pt minus 1pt}
\newcommand\afterAuthorSpace{\vskip5pt plus4pt minus4pt}
\newcommand\afterAffiliationSpace{\vskip3pt plus3pt}
\newcommand\afterEmailSpace{\vskip16pt plus9pt minus10pt\filbreak}
\newcommand\afterXtumSpace{\par\bigskip}
\newcommand\afterAbstractSpace{\vskip16pt plus9pt minus13pt}
\newcommand\afterKeywordsSpace{\vskip16pt plus9pt minus13pt}
\newcommand\afterArxivSpace{\vskip3pt plus0.01fil minus10pt}
\newcommand\afterDedicatedSpace{\vskip0pt plus0.01fil}
\newcommand\afterTocSpace{\bigskip\medskip}
\newcommand\afterTocRuleSpace{\bigskip\bigskip}
\newlength{\affiliationsSep}
\newcommand\beforetochook{\pagestyle{myplain}\pagenumbering{roman}}

\DeclareFixedFont\trfont{OT1}{phv}{b}{sc}{11}

\renewcommand\maketitle{
\pagestyle{empty}
\thispagestyle{titlepage}
\setcounter{page}{0}
\noindent{\small\scshape\@fpheader}\@preprint\par

\afterLogoSpace
\if!\@subheader!\else\noindent{\trfont{\@subheader}}\fi
\afterSubheaderSpace
\if!\@proceeding!\else\noindent{\sc\@proceeding}\fi
\afterProceedingsSpace
{\LARGE\flushleft\sffamily\bfseries\@title\par}
\afterTitleSpace
\hrule height 1.5\p@%
\afterRuleSpace
\if!\@collaboration!\else
{\Large\bfseries\sffamily\raggedright\@collaboration}\par
\afterCollaborationSpace
\fi
\if!\@collaborationImg!\else
{\normalsize\bfseries\sffamily\raggedright\@collaborationImg}\par
\afterCollaborationImgSpace
\fi
{\bfseries\raggedright\sffamily\the\auth@toks\par}
\afterAuthorSpace
\ifaffil\begin{list}{}{%
\setlength{\leftmargin}{0.28cm}%
\setlength{\labelsep}{0pt}%
\setlength{\itemsep}{\affiliationsSep}%
\setlength{\topsep}{-\parskip}}
\itshape\small%
\the\affil@toks
\end{list}\fi
\afterAffiliationSpace
\ifemailadd 
\noindent\hspace{0.28cm}\begin{minipage}[l]{.9\textwidth}
\begin{flushleft}
\textit{E-mail:} \the\email@toks
\end{flushleft}
\end{minipage}
\else 
\PackageWarningNoLine{\jname}{E-mails are missing.\MessageBreak Plese use \protect\emailAdd\space macro to provide e-mails.}
\fi
\afterEmailSpace
\if!\@xtum!\else\noindent{\@xtum}\afterXtumSpace\fi
\if!\@abstract!\else\noindent{\renewcommand\baselinestretch{.9}\textsc{Abstract:}}\ \@abstract\afterAbstractSpace\fi
\if!\@keywords!\else\noindent{\textsc{Keywords:}} \@keywords\afterKeywordsSpace\fi
\if!\@arxivnumber!\else\noindent{\textsc{ArXiv ePrint:}} \href{http://arxiv.org/abs/\@arxivnumber}{\@arxivnumber}\afterArxivSpace\fi
\if!\@dedicated!\else\vbox{\small\it\raggedleft\@dedicated}\afterDedicatedSpace\fi
\ifnotoc\else
\iftoccontinuous\else\newpage\fi
\beforetochook\hrule
\tableofcontents
\afterTocSpace
\hrule
\afterTocRuleSpace
\fi
\setcounter{footnote}{0}
\pagestyle{myplain}\pagenumbering{arabic}
} 

\renewcommand{\baselinestretch}{1.1}\normalsize
\divide\@tempdima\baselineskip
\@tempcnta=\@tempdima
\skip\@mpfootins = \skip\footins
\renewcommand{\@dotsep}{10000}

\newcommand\ps@myplain{
\pagenumbering{arabic}
\renewcommand\@oddfoot{\hfill-- \thepage\ --\hfill}
\renewcommand\@oddhead{}}
\let\ps@plain=\ps@myplain

\newcommand\ps@titlepage{\renewcommand\@oddfoot{}\renewcommand\@oddhead{}}

\numberwithin{equation}{section}

\renewcommand\section{\@startsection{section}{1}{\z@}%
                                   {-3.5ex \@plus -1.3ex \@minus -.7ex}%
                                   {2.3ex \@plus.4ex \@minus .4ex}%
                                   {\normalfont\large\bfseries}}
\renewcommand\subsection{\@startsection{subsection}{2}{\z@}%
                                   {-2.3ex\@plus -1ex \@minus -.5ex}%
                                   {1.2ex \@plus .3ex \@minus .3ex}%
                                   {\normalfont\normalsize\bfseries}}
\renewcommand\subsubsection{\@startsection{subsubsection}{3}{\z@}%
                                   {-2.3ex\@plus -1ex \@minus -.5ex}%
                                   {1ex \@plus .2ex \@minus .2ex}%
                                   {\normalfont\normalsize\bfseries}}
\renewcommand\paragraph{\@startsection{paragraph}{4}{\z@}%
                                   {1.75ex \@plus1ex \@minus.2ex}%
                                   {-1em}%
                                   {\normalfont\normalsize\bfseries}}
\renewcommand\subparagraph{\@startsection{subparagraph}{5}{\parindent}%
                                   {1.75ex \@plus1ex \@minus .2ex}%
                                   {-1em}%
                                   {\normalfont\normalsize\bfseries}}

\def\fnum@figure{\textbf{\figurename\nobreakspace\thefigure}}
\def\fnum@table{\textbf{\tablename\nobreakspace\thetable}}

\long\def\@makecaption#1#2{%
  \vskip\abovecaptionskip
  \sbox\@tempboxa{\small #1. #2}%
  \ifdim \wd\@tempboxa >\hsize
    \small #1. #2\par
  \else
    \global \@minipagefalse
    \hb@xt@\hsize{\hfil\box\@tempboxa\hfil}%
  \fi
  \vskip\belowcaptionskip}

\renewenvironment{thebibliography}[1]{%
\begin{oldthebibliography}{#1}%
\small%
\raggedright%
\setlength{\itemsep}{5pt plus 0.2ex minus 0.05ex}%
}%
{%
\end{oldthebibliography}%
}

\begin{document}


\title{\boldmath Heat-kernel approach for scattering}

\author[a]{Wen-Du Li}
\author[a,b,1]{and Wu-Sheng Dai}\note{daiwusheng@tju.edu.cn.}


\affiliation[a]{Department of Physics, Tianjin University, Tianjin 300072, P.R. China}
\affiliation[b]{LiuHui Center for Applied Mathematics, Nankai University \& Tianjin University, Tianjin 300072, P.R. China}








\abstract{An approach for solving scattering problems, based on two quantum field theory
methods, the heat kernel method and the scattering spectral method, is
constructed. This approach converts a method of calculating heat kernels into
a method of solving scattering problems. This allows us to establish a method
of scattering problems from a method of heat kernels. As an application, we
construct an approach for solving scattering problems based on the covariant
perturbation theory of heat-kernel expansions. In order to apply the
heat-kernel method to scattering problems, we first calculate the off-diagonal
heat-kernel expansion in the frame of the covariant perturbation theory.
Moreover, as an alternative application of the relation between heat kernels
and partial-wave phase shifts presented in this paper, we give an example of
how to calculate a global heat kernel from a known scattering phase shift.}

\maketitle
\flushbottom


\section{Introduction}

In this paper, based on two quantum field theory methods, heat-kernel method
\cite{vassilevich2003heat} and scattering spectral method
\cite{graham2009spectral}, we present a new approach to solve scattering
problems. This approach is a series of approaches for scatterings rather than
a single approach. Concretely, our key result is an explicit relation between
partial-wave scattering phase shifts and heat kernels. By this result, each
method of calculating heat kernels leads to an approach of calculating phase
shifts; or, in other words, the approach converts a method of solving heat
kernels into a method of solving scattering problems. Many methods for
scattering problems can be constructed by this approach, since the heat-kernel
theory is well studied in both mathematics and physics and there are many
mature\ methods for the calculation of heat kernels.

\textit{Phase shift.} All information of an elastic scattering process is
embedded in a scattering phase shift. This can be seen by directly observing
the asymptotic solution of the radial wave equation. For spherically symmetric
cases, the asymptotic solution of the free radial wave equation, $\left[
-\frac{1}{r^{2}}\frac{d}{dr}\left(  r^{2}\frac{d}{dr}\right)  +\frac{l\left(
l+1\right)  }{r^{2}}\right]  R_{l}=k^{2}R_{l}$, is $R_{l}\left(  r\right)
\overset{r\rightarrow\infty}{=}\left(  1/kr\right)  \sin\left(  kr-l\pi
/2\right)  $ and the asymptotic solution of the radial wave equation with a
potential, $\left[  -\frac{1}{r^{2}}\frac{d}{dr}\left(  r^{2}\frac{d}%
{dr}\right)  +\frac{l\left(  l+1\right)  }{r^{2}}+V\left(  r\right)  \right]
R_{l}=k^{2}R_{l}$, is
\begin{equation}
R_{l}\left(  r\right)  \overset{r\rightarrow\infty}{=}\frac{1}{kr}\sin\left[
kr-\frac{l\pi}{2}+\delta_{l}\left(  k\right)  \right]  . \label{Rrinf}%
\end{equation}
This defines the partial-wave phase shift $\delta_{l}\left(  k\right)  $,
which is the only effect on the radial wave function at asymptotic distances
\cite{ballentine1998quantum}. Therefore, all we need to do in solving a
scattering problem is to solve the phase shift $\delta_{l}\left(  k\right)  $.

\textit{Heat kernel. }The information embedded in an operator $D$ can be
extracted from a heat kernel $K\left(  t;\mathbf{r},\mathbf{r}^{\prime
}\right)  $ which is the Green function of the initial-value problem of the
heat-type equation $\left(  \partial_{t}+D\right)  \phi=0$, determined by
\cite{vassilevich2003heat}%
\begin{equation}
\left(  \partial_{t}+D\right)  K\left(  t;\mathbf{r},\mathbf{r}^{\prime
}\right)  =0,\text{ with }K\left(  0;\mathbf{r},\mathbf{r}^{\prime}\right)
=\delta\left(  \mathbf{r}-\mathbf{r}^{\prime}\right)  . \label{hkeq}%
\end{equation}
The global heat kernel $K\left(  t\right)  $ is the trace of the local heat
kernel $K\left(  t;\mathbf{r},\mathbf{r}^{\prime}\right)  $: $K\left(
t\right)  =\int d\mathbf{r}K\left(  t;\mathbf{r},\mathbf{r}\right)
=\sum_{n,l}e^{-\lambda_{nl}t}$, where $\lambda_{nl}$ is the eigenvalue of the
operator $D$.

The main aim of the present paper is to seek a relation between the
partial-wave phase shift $\delta_{l}\left(  k\right)  $ and\ the heat kernel
$K\left(  t;\mathbf{r},\mathbf{r}^{\prime}\right)  $. By this relation, we can
explicitly express a partial-wave phase shift by a given heat kernel. There
are many studies on the approximate calculation of heat kernels
\cite{vassilevich2003heat,culumovic1988calculation,dilkes1996off,mckeon1991seeley,kotani2000albanese,groh2011off,fliegner1994higher,barvinsky1987beyond,barvinsky1990covariant,barvinsky1990covariant3,mukhanov2007introduction}
and each approximate method of heat kernels gives us an approximate method for
calculating partial-wave phase shifts.

The present work is based on our preceding work given in Ref.
\cite{pang2012relation}, which reveals a relation between two quantum field
theory methods, the heat-kernel method \cite{vassilevich2003heat} and the
scattering spectral method \cite{graham2009spectral}. In Ref.
\cite{pang2012relation}, using the relation between spectral counting
functions and heat kernels given by Ref. \cite{dai2009number} and the relation
between phase shifts and state densities given by Ref.
\cite{graham2009spectral}, we provide a relation between the global heat
kernel and the total scattering phase shift (the total scattering phase shift
is the summation of all partial-wave phase shifts, $\delta\left(  k\right)
=\sum_{l}\left(  2l+1\right)  \delta_{l}\left(  k\right)  $).

Nevertheless, the result given by Ref. \cite{pang2012relation} --- the
relation between \textit{total} scattering phase shifts and heat kernels ---
can hardly be applied to scattering problems, since the \textit{total} phase
shift has no clear physical meaning.

To apply the heat-kernel method to scattering problems, we in fact need a
relation between partial-wave phase shifts (rather than \textit{total} phase
shifts) and heat kernels. In the present paper, we find such a relation. This
relation allows us to express a partial-wave phase shift by a known heat
kernel. Then all physical quantities of a scattering process, such as
scattering amplitudes and cross sections, can be expressed by a heat kernel.

To find the\ relation between partial-wave phase shifts and heat kernels, we
will first prove a relation between heat kernels and partial-wave heat
kernels. The heat kernel $K\left(  t;\mathbf{r},\mathbf{r}^{\prime}\right)  $
is the Green function of initial-value problem of the heat equation
(\ref{hkeq}) with the operator $D=-\nabla^{2}+V\left(  r\right)  $ and the
partial-wave heat kernel $K_{l}\left(  t;r,r^{\prime}\right)  $ is the Green
function of initial-value problem of the heat equation (\ref{hkeq}) with the
radial operator $D_{l}=-\frac{1}{r^{2}}\frac{d}{dr}\left(  r^{2}\frac{d}%
{dr}\right)  +\frac{l\left(  l+1\right)  }{r^{2}}+V\left(  r\right)  $. By
this relation, we can calculate a partial-wave heat kernel $K_{l}\left(
t;r,r^{\prime}\right)  $ from a heat kernel $K\left(  t;\mathbf{r}%
,\mathbf{r}^{\prime}\right)  $ directly.

The main aim of this paper is to explicitly express the partial-wave phase
shift by a given heat kernel. As mentioned above, by our result, each method
of calculating heat kernels can be converted to a method of calculating
scattering problems.

In order to calculate a scattering phase shift from a heat kernel, we need
off-diagonal heat kernels (i.e., heat kernels). For this purpose, in the
following, we first calculate an off-diagonal heat-kernel expansion in the
frame of the covariant perturbation theory. It should be pointed out that many
methods on the calculation of diagonal heat-kernel expansions in literature
can be directly apply to the calculation of off-diagonal heat kernels.

A method for calculating scattering phase shifts based on the covariant
perturbation theory in the heat-kernel theory is established as an example of
our approach.

Furthermore, we compare the scattering method established in this paper, which
is based on the covariant perturbation theory of heat kernels, with the Born
approximation. The comparison shows that the scattering method based on the
covariant perturbation theory is a better approximation than the Born approximation.

Besides applying the heat-kernel method to scattering problems, on the other
hand, by the method suggested in the present paper, we can also apply the
scattering method to the heat-kernel theory. In this paper, we provide a
simple example for illustrating how to calculate a heat kernel from a known
scattering result; more details on this subject will be given in a subsequent
work. The value of developing such a method, for example, is that though it is
relatively easy to obtain a high-energy expansion of heat kernels, it is
difficult to obtain a low-energy heat-kernel expansion. With the help of
scattering theory, we can calculate a low-energy heat-kernel expansion from a
low-energy scattering theory.

The starting point of this work, as mentioned above, is a relation between the
heat-kernel method and the scattering spectral method in quantum field theory.
The heat-kernel method is important both in physics and mathematics. In
physics, the heat-kernel method has important applications in, e.g., Euclidean
field theory, gravitation theory, and statistical mechanics
\cite{vassilevich2003heat,mukhanov2007introduction,dai2003quantum,dai2007interacting,dai2005hard}%
. In mathematics, the heat-kernel method is an important basis of the spectral
geometry \cite{vassilevich2003heat,davies1990heat}. There is much research on
the calculation of heat kernels. Besides exact solutions, there are many
systematic studies on the asymptotic expansion of heat kernels
\cite{bytsenko2003analytic}, such as the covariant perturbation theory
\cite{barvinsky1987beyond,barvinsky1990covariant,barvinsky1990covariant3,mukhanov2007introduction}%
. With various heat-kernel expansion techniques, one can obtain many
approximate solutions of heat kernels. Scattering spectral method is an
important quantum theory method which can be used to solve a variety of
problems in quantum field theory, e.g., to characterize the spectrum of energy
eigenstates in a potential background \cite{graham2009spectral} and to solve
the Casimir energy
\cite{graham2002calculating,farhi1998finite,farhi2000heavy,farhi2002searching,rahi2009scattering}%
. The method particular focuses on the property of the quantum vacuum.

In Sect. \ref{relation}, we find a relation between partial-wave phase shifts
and heat kernels. As a key step, we give a relation between partial-wave heat
kernels and heat kernels. In Sect. \ref{covariantperturbationps}, based on the
relation between partial-wave phase shifts and heat kernels given in Sect.
\ref{relation}, we establish an approach for the calculation of partial-wave
phase shifts, based on an heat-kernel expansion, the covariant perturbation
theory. In Sect. \ref{bornapp}, a comparison of the approach established in
the present paper and the Born approximation is given; in particular, we
compare these two methods through an exactly solvable potential.\ In Sect.
\ref{heatkernel}, we give an example for calculating a heat kernel from a
given phase shift. Conclusions and outlook are given in Sect. \ref{conclusion}%
. Moreover, an integral formula and two integral representations are given in
Appendices \ref{legendreintegral} and \ref{representationj2}.

\section{Relation between partial-wave phase shift and heat kernel:
calculating scattering phase shift from heat kernel \label{relation}}

The main result of the present paper is the following theorem which reveals a
relation between partial-wave scattering phase shifts and heat kernels. This
relation allows us to obtain a partial-wave phase shift from a known heat
kernel directly. By this relation, what we can obtain is not only one method
for scattering problems. It is in fact a series of methods for scattering
problems: each heat-kernel method leads to a method for solving scattering problems.

\begin{theorem}
\label{dltK} The relation between the partial-wave scattering phase shift,
$\delta_{l}\left(  k\right)  $, and the heat kernel, $K\left(  t;\mathbf{r}%
,\mathbf{r}^{\prime}\right)  =K\left(  t;r,\theta,\varphi,r^{\prime}%
,\theta^{\prime},\varphi^{\prime}\right)  $, is\
\begin{equation}
\delta_{l}\left(  k\right)  =2\pi^{2}\int_{0}^{\infty}r^{2}dr\int_{-1}%
^{1}d\cos\gamma P_{l}\left(  \cos\gamma\right)  \frac{1}{2\pi i}%
\int_{c-i\infty}^{c+i\infty}dt\frac{e^{k^{2}t}}{t}K^{s}\left(  t;r,\theta
,\varphi,r,\theta^{\prime},\varphi^{\prime}\right)  +\delta_{l}\left(
0\right)  , \label{hkps}%
\end{equation}
where $K^{s}\left(  t;\mathbf{r},\mathbf{r}^{\prime}\right)  $ is the
scattering part of a heat kernel, $P_{l}\left(  \cos\gamma\right)  $ is the
Legendre polynomial, and $\gamma$ is the angle between $\mathbf{r}$ and
$\mathbf{r}^{\prime}$ with $\cos\gamma=\cos\theta\cos\theta^{\prime}%
+\sin\theta\sin\theta^{\prime}\cos\left(  \varphi^{\prime}-\varphi\right)  $.
\end{theorem}

Notice that only the radial diagonal heat kernel, $K^{s}\left(  t;r,\theta
,\varphi,r,\theta^{\prime},\varphi^{\prime}\right)  $, appears in Eq.
(\ref{hkps}).\ The heat kernel $K\left(  t;\mathbf{r},\mathbf{r}^{\prime
}\right)  $ is split into three parts: $K\left(  t;\mathbf{r},\mathbf{r}%
^{\prime}\right)  =K^{s}\left(  t;\mathbf{r},\mathbf{r}^{\prime}\right)
+K^{b}\left(  t;\mathbf{r},\mathbf{r}^{\prime}\right)  +K^{f}\left(
t;\mathbf{r},\mathbf{r}^{\prime}\right)  $.\ The free part of a heat kernel
$K^{f}\left(  t;\mathbf{r},\mathbf{r}^{\prime}\right)  =$ $K^{\left(
0\right)  }\left(  t;\mathbf{r,r}^{\prime}\right)  $ is the heat kernel of the
operator $D=-\nabla^{2}$; the bound part of a heat kernel corresponds to the
bound state, if exists, of the system, which, in the spectral representation,
is $K^{b}\left(  t;\mathbf{r},\mathbf{r}^{\prime}\right)  =\sum_{\text{bound
states}}e^{-\lambda t}\psi_{\lambda}\left(  \mathbf{r}\right)  \psi_{\lambda
}^{\ast}\left(  \mathbf{r}^{\prime}\right)  $; the scattering part of a heat
kernel corresponds to the scattering state of the system, which, in the
spectral representation, is $K^{s}\left(  t;\mathbf{r},\mathbf{r}^{\prime
}\right)  =\sum_{\text{scattering states}}e^{-\lambda t}\psi_{\lambda}\left(
\mathbf{r}\right)  \psi_{\lambda}^{\ast}\left(  \mathbf{r}^{\prime}\right)  $
\cite{pang2012relation}. Note that $\delta_{l}\left(  0\right)  =\pi/2$ if
there is a half-bound state and $\delta_{l}\left(  0\right)  =0$ if there is
no half-bound state \cite{pang2012relation}.

The remaining task of this section is to prove this theorem. In order to prove
the theorem, we need to first find a relation between partial-wave heat
kernels and heat kernels.

\subsection{Relation between partial-wave heat kernel and heat kernel}

As mentioned above, the heat kernel $K\left(  t;\mathbf{r},\mathbf{r}^{\prime
}\right)  $ of an operator $D$ is determined by the heat equation (\ref{hkeq})
\cite{vassilevich2003heat}. For a spherically symmetric operator $D$, the heat
kernel $K\left(  t;\mathbf{r},\mathbf{r}^{\prime}\right)  =K\left(
t;r,\theta,\varphi,r^{\prime},\theta^{\prime},\varphi^{\prime}\right)  $ can
be expressed as
\begin{equation}
K\left(  t;\mathbf{r},\mathbf{r}^{\prime}\right)  =\sum_{n,l,m}e^{-\lambda
_{nl}t}\psi_{nlm}\left(  \mathbf{r}\right)  \psi_{nlm}^{\ast}\left(
\mathbf{r}^{\prime}\right)  , \label{hkdef}%
\end{equation}
where $\lambda_{nl}$ and $\psi_{nlm}\left(  \mathbf{r}\right)  =R_{nl}\left(
r\right)  Y_{lm}\left(  \theta,\varphi\right)  $ are the eigenvalue and
eigenfunction of $D$, determined by the eigenequation $D\psi_{nlm}%
=\lambda_{nl}\psi_{nlm}$, where $R_{nl}\left(  r\right)  $ is the radial wave
function and $Y_{lm}\left(  \theta,\varphi\right)  $ is the spherical
harmonics. The global heat kernel is the trace of the local heat kernel
$K\left(  t;\mathbf{r},\mathbf{r}^{\prime}\right)  $:%
\begin{equation}
K\left(  t\right)  =\int d\mathbf{r}K\left(  t;\mathbf{r},\mathbf{r}\right)
=\sum_{n,l}e^{-\lambda_{nl}t}.
\end{equation}

The local partial-wave heat kernel%
\begin{equation}
K_{l}\left(  t;r,r^{\prime}\right)  =\sum_{n}e^{-\lambda_{nl}t}R_{nl}\left(
r\right)  R_{nl}\left(  r^{\prime}\right)  \label{pwhkdef}%
\end{equation}
of the operator $D$ is the heat kernel of the $l$-th partial-wave radial
operator \cite{pang2012relation}%
\begin{equation}
D_{l}=-\frac{1}{r^{2}}\frac{d}{dr}\left(  r^{2}\frac{d}{dr}\right)
+\frac{l\left(  l+1\right)  }{r^{2}}+V\left(  r\right)
\end{equation}
which determines the radial equation $D_{l}R_{nl}=\lambda_{nl}R_{nl}$. The
global partial-wave heat kernel is the trace of the local partial-wave heat
kernel $K_{l}\left(  t;r,r^{\prime}\right)  $,
\begin{equation}
K_{l}\left(  t\right)  =\int_{0}^{\infty}r^{2}drK_{l}\left(  t;r,r\right)
=\sum_{n}e^{-\lambda_{nl}t}. \label{KltKltrr}%
\end{equation}

Now we prove that the relation between $K_{l}\left(  t;r,r^{\prime}\right)
$\ and $K\left(  t;\mathbf{r},\mathbf{r}^{\prime}\right)  $ can be expressed
as follows.

\begin{lemma}
\label{KKl}The relation between the partial-wave heat kernel $K_{l}\left(
t;r,r^{\prime}\right)  $ and the heat kernel $K\left(  t;\mathbf{r}%
,\mathbf{r}^{\prime}\right)  =K\left(  t;r,\theta,\varphi,r^{\prime}%
,\theta^{\prime},\varphi^{\prime}\right)  $ is%
\begin{equation}
K_{l}\left(  t;r,r^{\prime}\right)  =2\pi\int_{-1}^{1}d\cos\gamma P_{l}\left(
\cos\gamma\right)  K\left(  t;r,\theta,\varphi,r^{\prime},\theta^{\prime
},\varphi^{\prime}\right)  \label{pwhk-hk}%
\end{equation}
and%
\begin{equation}
K\left(  t;r,\theta,\varphi,r^{\prime},\theta^{\prime},\varphi^{\prime
}\right)  =\frac{1}{4\pi}\sum_{l}\left(  2l+1\right)  P_{l}\left(  \cos
\gamma\right)  K_{l}\left(  t;r,r^{\prime}\right)  . \label{hk-pwhk}%
\end{equation}

\end{lemma}

\begin{proof}
In a scattering with a spherically symmetric potential, the scattering wave
function $\psi_{nlm}\left(  r,\theta,\varphi\right)  =R_{nl}\left(  r\right)
Y_{lm}\left(  \theta,\varphi\right)  $. Then, by Eq. (\ref{hkdef}), the heat
kernel can be expressed as
\begin{equation}
K\left(  t;r,\theta,\varphi,r^{\prime},\theta^{\prime},\varphi^{\prime
}\right)  =\sum_{n,l}e^{-\lambda_{nl}t}R_{nl}\left(  r\right)  R_{nl}\left(
r^{\prime}\right)  \sum_{m=-l}^{l}Y_{lm}\left(  \theta,\varphi\right)
Y_{lm}^{\ast}\left(  \theta^{\prime},\varphi^{\prime}\right)  .
\end{equation}
Using the relation \cite{olver2010nist}
\begin{equation}
\sum_{m=-l}^{l}Y_{lm}\left(  \theta,\varphi\right)  Y_{lm}^{\ast}\left(
\theta^{\prime},\varphi^{\prime}\right)  =\frac{2l+1}{4\pi}P_{l}\left(
\cos\theta\cos\theta^{\prime}+\sin\theta\sin\theta^{\prime}\cos\left(
\varphi^{\prime}-\varphi\right)  \right)  ,
\end{equation}
we obtain%
\begin{equation}
K\left(  t;r,\theta,\varphi,r^{\prime},\theta^{\prime},\varphi^{\prime
}\right)  =\frac{1}{4\pi}\sum_{l}\left(  2l+1\right)  P_{l}\left(  \cos
\gamma\right)  \sum_{n}e^{-\lambda_{nl}t}R_{nl}\left(  r\right)  R_{nl}\left(
r^{\prime}\right)  .
\end{equation}
Then, by Eq. (\ref{pwhkdef}), we prove the relation (\ref{hk-pwhk}).

Multiplying both sides of Eq. (\ref{hk-pwhk}) by $P_{l^{\prime}}\left(
\cos\gamma\right)  $ and then integrating $\cos\gamma$ from $-1$ to $1$ give
\begin{align}
&  \int_{-1}^{1}d\cos\gamma P_{l^{\prime}}\left(  \cos\gamma\right)  K\left(
t;r,\theta,\varphi,r^{\prime},\theta^{\prime},\varphi^{\prime}\right)
\nonumber\\
&  =\frac{1}{4\pi}\sum_{l}\left(  2l+1\right)  \left[  \int_{-1}^{1}%
d\cos\gamma P_{l^{\prime}}\left(  \cos\gamma\right)  P_{l}\left(  \cos
\gamma\right)  \right]  K_{l}\left(  t;r,r^{\prime}\right)  .
\end{align}
Using the orthogonality of the Legendre polynomials \cite{olver2010nist}%
\begin{equation}
\int_{-1}^{1}d\cos\gamma P_{l^{\prime}}\left(  \cos\gamma\right)  P_{l}\left(
\cos\gamma\right)  =\frac{2}{2l^{\prime}+1}\delta_{ll^{\prime}},
\end{equation}
we obtain%
\begin{equation}
\int_{-1}^{1}d\cos\gamma P_{l^{\prime}}\left(  \cos\gamma\right)  K\left(
t;r,\theta,\varphi,r^{\prime},\theta^{\prime},\varphi^{\prime}\right)
=\frac{1}{2\pi}K_{l^{\prime}}\left(  t;r,r^{\prime}\right)  .
\end{equation}
This proves the relation (\ref{pwhk-hk}).
\end{proof}

\subsection{Proof of Theorem \ref{dltK}}

Now, with Lemma \ref{KKl}, we can prove Theorem \ref{dltK}.

\begin{proof}
In Ref. \cite{pang2012relation}, we prove a relation between total phase
shifts and global heat kernels,%
\begin{equation}
\delta\left(  k\right)  =\frac{1}{2i}\int_{c-i\infty}^{c+i\infty}\frac
{K^{s}\left(  t\right)  }{t}e^{k^{2}t}dt+\delta\left(  0\right)  ,
\label{totps}%
\end{equation}
and a relation between partial-wave phase shifts and partial-wave global heat
kernels,%
\begin{equation}
\delta_{l}\left(  k\right)  =\frac{1}{2i}\int_{c-i\infty}^{c+i\infty}%
\frac{K_{l}^{s}\left(  t\right)  }{t}e^{k^{2}t}dt+\delta_{l}\left(  0\right)
. \label{psandpwhk}%
\end{equation}
Here the global heat kernel and the global partial-wave heat kernel are split
into the scattering part, the bound part, and the free part: $K\left(
t\right)  =K^{s}\left(  t\right)  +K^{b}\left(  t\right)  +K^{f}\left(
t\right)  $ and $K_{l}\left(  t\right)  =K_{l}^{s}\left(  t\right)  +K_{l}%
^{b}\left(  t\right)  +K_{l}^{f}\left(  t\right)  $ \cite{pang2012relation}.

Starting from the global partial-wave heat kernel given by Eq. (\ref{KltKltrr}%
) and using\ the relation between partial-wave heat kernels and heat kernels
given by Lemma \ref{KKl}, Eq. (\ref{pwhk-hk}), we have
\begin{align}
K_{l}\left(  t\right)   &  =\int_{0}^{\infty}r^{2}drK_{l}\left(  t;r,r\right)
\nonumber\\
&  =2\pi\int_{0}^{\infty}r^{2}dr\int_{-1}^{1}d\cos\gamma P_{l}\left(
\cos\gamma\right)  K\left(  t;r,\theta,\varphi,r,\theta^{\prime}%
,\varphi^{\prime}\right)  . \label{traceKlt}%
\end{align}
Substituting Eq. (\ref{traceKlt}) into Eq. (\ref{psandpwhk}) proves Theorem
\ref{dltK}.
\end{proof}

It should be noted here that the relations given by Ref.
\cite{pang2012relation}, Eqs. (\ref{totps}) and (\ref{psandpwhk}),\ only allow
one to calculate the \textit{total} phase shift $\delta\left(  k\right)  $
from a heat kernel $K\left(  t\right)  $ or to calculate the partial-wave
phase shift $\delta_{l}\left(  k\right)  $ from a partial-wave heat kernel
$K_{l}\left(  t\right)  $. Such results, however, are not useful in scattering
problems, because the total phase shift $\delta\left(  k\right)  $ is not
physically meaningful and the partial-wave heat kernel $K_{l}\left(  t\right)
$ is often difficult to obtain.

Nevertheless, the result given by Theorem \ref{dltK}, Eq. (\ref{hkps}), allows
one to calculate the partial-wave phase shift $\delta_{l}\left(  k\right)  $
from a heat kernel $K\left(  t\right)  $ rather than a partial-wave heat
kernel $K_{l}\left(  t\right)  $. The heat kernel has been fully studied and
there are many well-known results \cite{vassilevich2003heat}.

\section{Heat-kernel approach for phase shift: covariant perturbation theory
\label{covariantperturbationps}}

In this section, based on the heat-kernel expansion given by the covariant
perturbation theory
\cite{barvinsky1987beyond,barvinsky1990covariant,barvinsky1990covariant3}, by
the relation between partial-wave phase shifts and heat kernels given by Eq.
(\ref{hkps}), we establish an expansion for scattering phase shifts. The
covariant perturbation theory is suitable for our purposes, since it provides
a uniformly convergent expansion of heat kernels
\cite{barvinsky1987beyond,gusev2009heat}.

\textit{The covariant perturbation theory type expansion for a partial-wave
scattering phase shift is }$\delta_{l}\left(  k\right)  =\delta_{l}^{\left(
1\right)  }\left(  k\right)  +\delta_{l}^{\left(  2\right)  }\left(  k\right)
+\cdots$ with%
\begin{align}
\delta_{l}^{\left(  1\right)  }\left(  k\right)   &  =-\frac{\pi}{2}\int%
_{0}^{\infty}rdrV\left(  r\right)  J_{l+1/2}^{2}\left(  kr\right)
,\label{xiebianps1}\\
\delta_{l}^{\left(  2\right)  }\left(  k\right)   &  =-\frac{\pi^{2}}{2}%
\int_{0}^{\infty}rdrJ_{l+1/2}\left(  kr\right)  Y_{l+1/2}\left(  kr\right)
V\left(  r\right)  \int_{0}^{r}r^{\prime}dr^{\prime}J_{l+1/2}^{2}\left(
kr^{\prime}\right)  V\left(  r^{\prime}\right)  , \label{xiebianps2}%
\end{align}
\textit{where }$J_{\nu}\left(  z\right)  $ and $Y_{\nu}\left(  z\right)
$\textit{ are the Bessel functions of the first and second kinds, respectively
\cite{olver2010nist}.}

A detailed calculation is as follows.

\subsection{Covariant perturbation theory for heat-kernel expansion}

The heat-kernel expansion is systematically studied in the covariant
perturbation theory
\cite{barvinsky1987beyond,barvinsky1990covariant,barvinsky1990covariant3}. The
heat-kernel expansion given by covariant perturbation theory reads
\cite{culumovic1988calculation,mckeon1991seeley}
\begin{align}
K\left(  t;\mathbf{r,r}^{\prime}\right)   &  =K^{\left(  0\right)  }\left(
t;\mathbf{r},\mathbf{r}^{\prime}\right)  +K^{\left(  1\right)  }\left(
t;\mathbf{r},\mathbf{r}^{\prime}\right)  +K^{\left(  2\right)  }\left(
t;\mathbf{r,r}^{\prime}\right)  +\cdots\nonumber\\
&  =\left\langle \mathbf{r}\right\vert e^{-H_{0}t}+\left(  -t\right)  \int%
_{0}^{\infty}d\alpha_{1}d\alpha_{2}\delta\left(  1-\alpha_{1}-\alpha
_{2}\right)  e^{-\alpha_{1}H_{0}t}Ve^{-\alpha_{2}H_{0}t}\nonumber\\
&  +\left(  -t\right)  ^{2}\int_{0}^{\infty}d\alpha_{1}d\alpha_{2}d\alpha
_{3}\delta\left(  1-\alpha_{1}-\alpha_{2}-\alpha_{3}\right)  e^{-\alpha
_{1}H_{0}t}Ve^{-\alpha_{2}H_{0}t}Ve^{-\alpha_{3}H_{0}t}+\cdots\left\vert
\mathbf{r}^{\prime}\right\rangle ,\label{hkcpt}%
\end{align}
where%
\begin{equation}
K^{\left(  0\right)  }\left(  t;\mathbf{r,r}^{\prime}\right)  =\left\langle
\mathbf{r}\right\vert e^{-H_{0}t}\left\vert \mathbf{r}^{\prime}\right\rangle
=\frac{1}{\left(  4\pi t\right)  ^{3/2}}e^{-\left\vert \mathbf{r}%
-\mathbf{r}^{\prime}\right\vert ^{2}/\left(  4t\right)  }\label{hkfree}%
\end{equation}
is the zero-order (free) heat kernel. Substituting the zero-order heat kernel
(\ref{hkfree}) into Eq. (\ref{hkcpt}), we obtain the first two orders of a
heat kernel,
\begin{align}
K^{\left(  1\right)  }\left(  t;\mathbf{r},\mathbf{r}^{\prime}\right)   &
=\left\langle \mathbf{r}\right\vert \left(  -t\right)  \int_{0}^{\infty
}d\alpha_{1}d\alpha_{2}\delta\left(  1-\alpha_{1}-\alpha_{2}\right)
e^{-\alpha_{1}H_{0}t}Ve^{-\alpha_{2}H_{0}t}\left\vert \mathbf{r}^{\prime
}\right\rangle \nonumber\\
&  =-\int_{0}^{t}d\tau\int d^{3}yV\left(  \mathbf{y}\right)  \frac{\exp\left(
-\frac{1}{4\left(  t-\tau\right)  }\left\vert \mathbf{r-y}\right\vert
^{2}\right)  }{\left[  4\pi\left(  t-\tau\right)  \right]  ^{3/2}}\frac
{\exp\left(  -\frac{1}{4\tau}\left\vert \mathbf{y-r}^{\prime}\right\vert
^{2}\right)  }{\left(  4\pi\tau\right)  ^{3/2}}\label{k1xiebian}%
\end{align}
and%
\begin{align}
K^{\left(  2\right)  }\left(  t;\mathbf{r,r}^{\prime}\right)   &
=\left\langle \mathbf{r}\right\vert \left(  -t\right)  ^{2}\int_{0}^{\infty
}d\alpha_{1}d\alpha_{2}d\alpha_{3}\delta\left(  1-\alpha_{1}-\alpha_{2}%
-\alpha_{3}\right)  e^{-\alpha_{1}H_{0}t}Ve^{-\alpha_{2}H_{0}t}Ve^{-\alpha
_{3}H_{0}t}\left\vert \mathbf{r}^{\prime}\right\rangle \nonumber\\
&  =\int d^{3}yV\left(  \mathbf{y}\right)  \int d^{3}zV\left(  \mathbf{z}%
\right)  \int_{0}^{t}d\tau\int_{0}^{\tau}d\tau^{\prime}\nonumber\\
&  \times\frac{\exp\left(  -\frac{1}{4\left(  t-\tau\right)  }\left\vert
\mathbf{r}-\mathbf{y}\right\vert ^{2}\right)  }{\left[  4\pi\left(
t-\tau\right)  \right]  ^{3/2}}\frac{\exp\left(  -\frac{1}{4\left(  \tau
-\tau^{\prime}\right)  }\left\vert \mathbf{y}-\mathbf{z}\right\vert
^{2}\right)  }{\left[  4\pi\left(  \tau-\tau^{\prime}\right)  \right]  ^{3/2}%
}\frac{\exp\left(  -\frac{1}{4\tau^{\prime}}\left\vert \mathbf{z}%
-\mathbf{r}^{\prime}\right\vert ^{2}\right)  }{\left(  4\pi\tau^{\prime
}\right)  ^{3/2}}.\label{k2xiebian}%
\end{align}
For the spherical potentials $V\left(  \mathbf{r}\right)  =V\left(  r\right)
$, $K^{\left(  1\right)  }\left(  t;\mathbf{r},\mathbf{r}^{\prime}\right)  $
and $K^{\left(  2\right)  }\left(  t;\mathbf{r,r}^{\prime}\right)  $ given by
Eqs. (\ref{k1xiebian}) and (\ref{k2xiebian}) become
\begin{align}
K^{\left(  1\right)  }\left(  t;r,r^{\prime},\gamma\right)   &  =-\int%
_{0}^{\infty}y^{2}dyV\left(  y\right)  \int d\Omega_{y}\int_{0}^{t}%
d\tau\nonumber\\
&  \times\frac{\exp\left(  -\frac{1}{4\left(  t-\tau\right)  }\left(
r^{2}+y^{2}-2ry\cos\gamma_{\mathbf{ry}}\right)  \right)  }{\left[  4\pi\left(
t-\tau\right)  \right]  ^{3/2}}\frac{\exp\left(  -\frac{1}{4\tau}\left(
r^{\prime2}+y^{2}-2r^{\prime}y\cos\gamma_{\mathbf{r}^{\prime}\mathbf{y}%
}\right)  \right)  }{\left(  4\pi\tau\right)  ^{3/2}}\label{k1xiebiansph}%
\end{align}
and%
\begin{align}
K^{\left(  2\right)  }\left(  t;r,r^{\prime},\gamma\right)   &  =\int%
_{0}^{\infty}y^{2}dyV\left(  y\right)  \int d\Omega_{y}\int_{0}^{\infty}%
z^{2}dzV\left(  z\right)  \nonumber\\
&  \times\int d\Omega_{z}\int_{0}^{t}d\tau\int_{0}^{\tau}d\tau^{\prime}%
\frac{\exp\left(  -\frac{1}{4\left(  t-\tau\right)  }\left(  r^{2}%
+y^{2}-2ry\cos\gamma_{\mathbf{ry}}\right)  \right)  }{\left[  4\pi\left(
t-\tau\right)  \right]  ^{3/2}}\\
&  \times\frac{\exp\left(  -\frac{1}{4\left(  \tau-\tau^{\prime}\right)
}\left(  y^{2}+z^{2}-2yz\cos\gamma_{\mathbf{yz}}\right)  \right)  }{\left[
4\pi\left(  \tau-\tau^{\prime}\right)  \right]  ^{3/2}}\frac{\exp\left(
-\frac{1}{4\tau^{\prime}}\left(  z^{2}+r^{\prime2}-2zr^{\prime}\cos
\gamma_{\mathbf{zr}^{\prime}}\right)  \right)  }{\left(  4\pi\tau^{\prime
}\right)  ^{3/2}},\label{k2xiebiansph}%
\end{align}
where $\gamma$ is the angle between $\mathbf{r}$ and $\mathbf{r}^{\prime}$,
$\gamma_{\mathbf{ry}}$ is the angle between $\mathbf{r}$ and $\mathbf{y}$,
$\gamma_{\mathbf{r}^{\prime}\mathbf{y}}$ is the angle between $\mathbf{r}%
^{\prime}$ and $\mathbf{y}$, $\gamma_{\mathbf{yz}}$ is the angle between
$\mathbf{y}$ and $\mathbf{z}$, and $\gamma_{\mathbf{zr}^{\prime}}$ is the
angle between $\mathbf{z}$ and $\mathbf{r}^{\prime}$.

\subsection{First-order phase shift $\delta_{l}^{\left(  1\right)  }\left(
k\right)  $}

In this section, we calculate the first-order phase shift in the frame of the
covariant perturbation theory.

The first-order phase shift $\delta_{l}^{\left(  1\right)  }\left(  k\right)
$ can be obtained by substituting the first-order heat kernel given by
covariant perturbation theory, Eq. (\ref{k1xiebiansph}), into the relation
between partial-wave phase shifts and heat kernels, Eq. (\ref{hkps}), and
taking $r^{\prime}=r$ (radial diagonal):%
\begin{equation}
\delta_{l}^{\left(  1\right)  }\left(  k\right)  =-2\pi^{2}\int_{0}^{\infty
}r^{2}dr\frac{1}{2\pi i}\int_{c-i\infty}^{c+i\infty}dt\frac{e^{k^{2}t}}{t}%
\int_{0}^{t}d\tau\int_{0}^{\infty}y^{2}dyV\left(  y\right)  \frac{\exp\left(
-\frac{r^{2}+y^{2}}{4\left(  t-\tau\right)  }\right)  }{\left[  4\pi\left(
t-\tau\right)  \right]  ^{3/2}}\frac{\exp\left(  -\frac{r^{2}+y^{2}}{4\tau
}\right)  }{\left(  4\pi\tau\right)  ^{3/2}}\mathcal{I}_{1},\label{deltaIry}%
\end{equation}
where $\mathcal{I}_{1}$ is an integral with respect to the angle,%
\begin{equation}
\mathcal{I}_{1}=\int_{-1}^{1}d\cos\gamma P_{l}\left(  \cos\gamma\right)  \int
d\Omega_{y}\exp\left(  \frac{ry}{2\left(  t-\tau\right)  }\cos\gamma
_{\mathbf{ry}}\right)  \exp\left(  \frac{ry}{2\tau}\cos\gamma_{\mathbf{r}%
^{\prime}\mathbf{y}}\right)  .\label{jiaodujifen1}%
\end{equation}
To calculate $\mathcal{I}_{1}$, we use the expansion \cite{olver2010nist}%
\begin{equation}
e^{iz\cos a}=\sum_{l=0}^{\infty}\left(  2l+1\right)  i^{l}\sqrt{\frac{\pi}%
{2z}}J_{l+1/2}\left(  z\right)  P_{l}\left(  \cos a\right)  \label{eizcosa}%
\end{equation}
to rewrite $\mathcal{I}_{1}$ as%
\begin{align}
\mathcal{I}_{1} &  =\sum_{l_{1}=0}\left(  2l_{1}+1\right)  i^{l_{1}}j_{l_{1}%
}\left(  \frac{ry}{i2\left(  t-\tau\right)  }\right)  \sum_{l_{2}=0}\left(
2l_{2}+1\right)  i^{l_{2}}j_{l_{2}}\left(  \frac{ry}{i2\tau}\right)
\nonumber\\
&  \times\int_{-1}^{1}d\cos\gamma P_{l}\left(  \cos\gamma\right)  \int
d\Omega_{y}P_{l_{1}}\left(  \cos\gamma_{\mathbf{ry}}\right)  P_{l_{2}}\left(
\cos\gamma_{\mathbf{r}^{\prime}\mathbf{y}}\right)  ,
\end{align}
where $j_{\nu}\left(  z\right)  =\sqrt{\pi/\left(  2z\right)  }J_{\nu
+1/2}\left(  z\right)  $ is the spherical Bessel function of the first kind.
Without loss of generality, we choose $\mathbf{r}^{\prime}=\left(  r^{\prime
},0,0\right)  $ and then $\gamma_{\mathbf{r}^{\prime}\mathbf{y}}=\theta_{y}$
and $\gamma=\theta_{r}$. Now, the integral with respect to $\Omega_{y}$ can be
worked out directly by using the integral formula (\ref{Lintegral}) given in
\ref{legendreintegral}:%
\begin{equation}
\int d\Omega_{y}P_{l_{1}}\left(  \cos\gamma_{\mathbf{ry}}\right)  P_{l_{2}%
}\left(  \cos\gamma_{\mathbf{r}^{\prime}\mathbf{y}}\right)  =\int d\Omega
_{y}P_{l_{1}}\left(  \cos\gamma_{\mathbf{ry}}\right)  P_{l_{2}}\left(
\cos\theta_{y}\right)  =P_{l_{1}}\left(  \cos\theta_{r}\right)  \frac{4\pi
}{2l_{1}+1}\delta_{l_{1},l_{2}}.\label{jiaodujifen11}%
\end{equation}
The integral with respect to $\gamma$ ($=\theta_{r}$), then, can also be
worked out by using the orthogonality of the Legendre polynomials $\int%
_{-1}^{1}dxP_{l^{\prime}}\left(  x\right)  P_{l}\left(  x\right)  =2/\left(
2l^{\prime}+1\right)  \delta_{ll^{\prime}}$ \cite{olver2010nist}:
\begin{equation}
\int_{-1}^{1}d\cos\theta_{r}P_{l}\left(  \cos\theta_{r}\right)  P_{l_{1}%
}\left(  \cos\theta_{r}\right)  \frac{4\pi}{2l_{1}+1}\delta_{l_{1},l_{2}%
}=\frac{8\pi}{\left(  2l+1\right)  ^{2}}\delta_{l,l_{1}}\delta_{l_{1},l_{2}%
}.\label{jiaodujifen12}%
\end{equation}
By Eqs. (\ref{jiaodujifen11}) and \ (\ref{jiaodujifen12}), we achieve%
\begin{align}
\mathcal{I}_{1} &  =\sum_{l_{1}=0}\left(  2l_{1}+1\right)  i^{l_{1}}j_{l_{1}%
}\left(  \frac{ry}{i2\left(  t-\tau\right)  }\right)  \sum_{l_{2}=0}\left(
2l_{2}+1\right)  i^{l_{2}}j_{l_{2}}\left(  \frac{ry}{i2\tau}\right)
\frac{8\pi}{\left(  2l+1\right)  ^{2}}\delta_{l,l_{1}}\delta_{l_{1},l_{2}%
}\nonumber\\
&  =8\pi i^{2l}j_{l}\left(  \frac{ry}{i2\left(  t-\tau\right)  }\right)
j_{l}\left(  \frac{ry}{i2\tau}\right)  .\label{Iry}%
\end{align}
Substituting Eq. (\ref{Iry}) into Eq. (\ref{deltaIry}) gives%
\begin{equation}
\delta_{l}^{\left(  1\right)  }\left(  k\right)  =-2\pi^{2}\frac{1}{2\pi
i}\int_{c-i\infty}^{c+i\infty}dt\frac{e^{k^{2}t}}{t}\int_{0}^{t}d\tau\int%
_{0}^{\infty}y^{2}dyV\left(  y\right)  \mathcal{I}_{2},\label{xbps1}%
\end{equation}
where
\begin{equation}
\mathcal{I}_{2}=8\pi i^{2l}\int_{0}^{\infty}r^{2}dr\frac{\exp\left(
-\frac{r^{2}+y^{2}}{4\left(  t-\tau\right)  }\right)  }{\left[  4\pi\left(
t-\tau\right)  \right]  ^{3/2}}\frac{\exp\left(  -\frac{r^{2}+y^{2}}{4\tau
}\right)  }{\left(  4\pi\tau\right)  ^{3/2}}j_{l}\left(  \frac{ry}{i2\left(
t-\tau\right)  }\right)  j_{l}\left(  \frac{ry}{i2\tau}\right)
.\label{Irgamma}%
\end{equation}
To calculate the integral $\mathcal{I}_{2}$, we use the integral
representation, Eq. (\ref{j2}), given in \ref{representationj2} to represent
the factor $j_{l}\left(  \frac{ry}{i2\left(  t-\tau\right)  }\right)
j_{l}\left(  \frac{ry}{i2\tau}\right)  $ as%
\begin{equation}
j_{l}\left(  \frac{ry}{i2\left(  t-\tau\right)  }\right)  j_{l}\left(
\frac{ry}{i2\tau}\right)  =\frac{1}{2}\int_{-1}^{1}d\cos\theta\frac{\sin
\sqrt{\left[  \frac{ry}{i2\left(  t-\tau\right)  }\right]  ^{2}+\left(
\frac{ry}{i2\tau}\right)  ^{2}-2\frac{ry}{i2\left(  t-\tau\right)  }\frac
{ry}{i2\tau}\cos\theta}}{\sqrt{\left[  \frac{ry}{i2\left(  t-\tau\right)
}\right]  ^{2}+\left(  \frac{ry}{i2\tau}\right)  ^{2}-2\frac{ry}{i2\left(
t-\tau\right)  }\frac{ry}{i2\tau}\cos\theta}}P_{l}\left(  \cos\theta\right)
.\label{doublej}%
\end{equation}
Substituting the integral representation (\ref{doublej}) into Eq.
(\ref{Irgamma}) and working out the integral give%
\begin{align}
\mathcal{I}_{2} &  =4\pi i^{2l}\int_{-1}^{1}d\cos\theta P_{l}\left(
\cos\theta\right)  \int_{0}^{\infty}r^{2}dr\frac{\exp\left(  -\frac
{r^{2}+y^{2}}{4\left(  t-\tau\right)  }\right)  }{\left[  4\pi\left(
t-\tau\right)  \right]  ^{3/2}}\frac{\exp\left(  -\frac{r^{2}+y^{2}}{4\tau
}\right)  }{\left(  4\pi\tau\right)  ^{3/2}}\nonumber\\
&  \times\frac{\sin\sqrt{\left[  \frac{ry}{i2\left(  t-\tau\right)  }\right]
^{2}+\left(  \frac{ry}{i2\tau}\right)  ^{2}-2\frac{ry}{i2\left(
t-\tau\right)  }\frac{ry}{i2\tau}\cos\theta}}{\sqrt{\left[  \frac
{ry}{i2\left(  t-\tau\right)  }\right]  ^{2}+\left(  \frac{ry}{i2\tau}\right)
^{2}-2\frac{ry}{i2\left(  t-\tau\right)  }\frac{ry}{i2\tau}\cos\theta}}\\
&  =\frac{i^{2l}}{\left(  4\pi t\right)  ^{3/2}}\int_{-1}^{1}d\cos\theta
P_{l}\left(  \cos\theta\right)  \exp\left(  -\frac{y^{2}}{2t}\left(
\cos\theta+1\right)  \right)  .\label{Irgamma1}%
\end{align}
Substituting Eq. (\ref{Irgamma1}) into Eq. (\ref{xbps1}) and performing the
integral with respect to $\tau$, we have%
\begin{equation}
\delta_{l}^{\left(  1\right)  }\left(  k\right)  =-\frac{\sqrt{\pi}}{4}%
i^{2l}\int_{0}^{\infty}y^{2}dyV\left(  y\right)  \frac{1}{2\pi i}%
\int_{c-i\infty}^{c+i\infty}dt\frac{e^{k^{2}t}}{t}\frac{1}{\sqrt{t}}%
e^{-y^{2}/\left(  2t\right)  }\int_{-1}^{1}d\cos\theta P_{l}\left(  \cos
\theta\right)  \exp\left(  -\frac{y^{2}\cos\theta}{2t}\right)
.\label{deltaintgamma}%
\end{equation}
Using the expansion $\exp\left(  -y^{2}\cos\theta/\left(  2t\right)  \right)
=\sum_{l=0}\left(  2l+1\right)  i^{l}j_{l}\left(  i\frac{y^{2}}{2t}\right)
P_{l}\left(  \cos\theta\right)  $ (see Eq. (\ref{eizcosa})) and the
orthogonality of the Legendre polynomials, we can work out the integral:%
\begin{align}
\int_{-1}^{1}d\cos\theta P_{l}\left(  \cos\theta\right)  \exp\left(
-\frac{y^{2}\cos\theta}{2t}\right)   &  =\sum_{l^{\prime}=0}\left(
2l^{\prime}+1\right)  i^{l^{\prime}}j_{l^{\prime}}\left(  i\frac{y^{2}}%
{2t}\right)  \int_{-1}^{1}d\cos\theta P_{l}\left(  \cos\theta\right)
P_{l^{\prime}}\left(  \cos\theta\right)  \nonumber\\
&  =2i^{l}j_{l}\left(  i\frac{y^{2}}{2t}\right)  =i^{2l}\frac{2\sqrt{\pi t}%
}{y}I_{l+1/2}\left(  \frac{y^{2}}{2t}\right)  ,\label{intgamma}%
\end{align}
where $I_{v}\left(  z\right)  $ is the modified Bessel function of the first
kind and the relation $j_{\nu}\left(  z\right)  =\sqrt{\pi/\left(  2z\right)
}i^{l}I_{\nu+1/2}\left(  z\right)  $ is used. Substituting Eq. (\ref{intgamma}%
) into Eq. (\ref{deltaintgamma}), we have%
\begin{equation}
\delta_{l}^{\left(  1\right)  }\left(  k\right)  =-\frac{\pi}{2}\int%
_{0}^{\infty}ydyV\left(  y\right)  \frac{1}{2\pi i}\int_{c-i\infty}%
^{c+i\infty}dt\frac{e^{k^{2}t}}{t}e^{-y^{2}/\left(  2t\right)  }%
I_{l+1/2}\left(  \frac{y^{2}}{2t}\right)  .\label{xbps2}%
\end{equation}
Finally, by performing the inverse Laplace transformation in Eq.
(\ref{xbps2}),
\begin{equation}
\frac{1}{2\pi i}\int_{c-i\infty}^{c+i\infty}dt\frac{e^{k^{2}t}}{t}%
e^{-r^{2}/\left(  2t\right)  }I_{l+1/2}\left(  \frac{r^{2}}{2t}\right)
=J_{l+1/2}^{2}\left(  kr\right)  ,
\end{equation}
the first-order phase shift given by covariant perturbation theory, Eq.
(\ref{xiebianps1}), is obtained.

\subsection{Second-order phase shift $\delta_{l}^{\left(  2\right)  }\left(
k\right)  $}

In this section, we calculate the second-order phase shift in the frame of the
covariant perturbation theory.

The second-order phase shift $\delta_{l}^{\left(  2\right)  }\left(  k\right)
$ can be obtained by substituting the second-order heat kernel given by
covariant perturbation theory, Eq. (\ref{k2xiebiansph}), into the relation
between partial-wave phase shifts and heat kernels, Eq. (\ref{hkps}), and
taking $r^{\prime}=r$:%
\begin{align}
\delta_{l}^{\left(  2\right)  }\left(  k\right)   &  =2\pi^{2}\int_{0}%
^{\infty}r^{2}dr\frac{1}{2\pi i}\int_{c-i\infty}^{c+i\infty}dt\frac{e^{k^{2}%
t}}{t}\int_{0}^{t}d\tau\int_{0}^{\tau}d\tau^{\prime}\nonumber\\
&  \times\int_{0}^{\infty}y^{2}dyV\left(  y\right)  \int_{0}^{\infty}%
z^{2}dzV\left(  z\right)  \frac{\exp\left(  -\frac{r^{2}+y^{2}}{4\left(
t-\tau\right)  }\right)  }{\left[  4\pi\left(  t-\tau\right)  \right]  ^{3/2}%
}\frac{\exp\left(  -\frac{y^{2}+z^{2}}{4\left(  \tau-\tau^{\prime}\right)
}\right)  }{\left[  4\pi\left(  \tau-\tau^{\prime}\right)  \right]  ^{3/2}%
}\frac{\exp\left(  -\frac{z^{2}+r^{2}}{4\tau^{\prime}}\right)  }{\left(
4\pi\tau^{\prime}\right)  ^{3/2}}\mathcal{I}_{3},\label{deltaIryz}%
\end{align}
where%
\begin{equation}
\mathcal{I}_{3}=\int_{-1}^{1}d\cos\gamma P_{l}\left(  \cos\gamma\right)  \int
d\Omega_{y}\int d\Omega_{z}\exp\left(  \frac{ry\cos\gamma_{\mathbf{ry}}%
}{2\left(  t-\tau\right)  }\right)  \exp\left(  \frac{yz\cos\gamma
_{\mathbf{yz}}}{2\left(  \tau-\tau^{\prime}\right)  }\right)  \exp\left(
\frac{zr\cos\gamma_{\mathbf{zr}^{\prime}}}{2\tau^{\prime}}\right)  .
\end{equation}
Using Eq. (\ref{eizcosa}), we rewrite the integral $\mathcal{I}_{3}$ as%
\begin{align}
\mathcal{I}_{3} &  =\sum_{l_{1}=0}\left(  2l_{1}+1\right)  i^{l_{1}}j_{l_{1}%
}\left(  \frac{ry}{i2\left(  t-\tau\right)  }\right)  \sum_{l_{2}=0}\left(
2l_{2}+1\right)  i^{l_{2}}j_{l_{2}}\left(  \frac{yz}{i2\left(  \tau
-\tau^{\prime}\right)  }\right)  \sum_{l_{3}=0}\left(  2l_{3}+1\right)
i^{l_{3}}j_{l_{3}}\left(  \frac{zr}{i2\tau^{\prime}}\right)  \nonumber\\
&  \times\int_{-1}^{1}d\cos\gamma P_{l}\left(  \cos\gamma\right)  \int
d\Omega_{y}\int d\Omega_{z}P_{l_{1}}\left(  \cos\gamma_{\mathbf{ry}}\right)
P_{l_{2}}\left(  \cos\gamma_{\mathbf{yz}}\right)  P_{l_{3}}\left(  \cos
\gamma_{\mathbf{zr}^{\prime}}\right)  .\label{Iryz}%
\end{align}
Without loss of generality, we choose $\mathbf{r}^{\prime}=\left(  r^{\prime
},0,0\right)  $ and then we have $\gamma_{\mathbf{zr}^{\prime}}=\theta_{z}$.
The integral with respect to $\Omega_{z}$ can then be worked out by use of the
integral formula, Eq. (\ref{Lintegral}), given in Appendix
\ref{legendreintegral}:%
\begin{equation}
\int d\Omega_{z}P_{l_{2}}\left(  \cos\gamma_{\mathbf{yz}}\right)  P_{l_{3}%
}\left(  \cos\gamma_{\mathbf{zr}^{\prime}}\right)  =\int d\Omega_{z}P_{l_{2}%
}\left(  \cos\gamma_{\mathbf{yz}}\right)  P_{l_{3}}\left(  \cos\theta
_{z}\right)  =P_{l_{2}}\left(  \cos\theta_{y}\right)  \frac{4\pi}{2l_{2}%
+1}\delta_{l_{2},l_{3}},\label{intIryz1}%
\end{equation}
The integral with respect to $\Omega_{y}$ also can be integrated directly by
Eq. (\ref{Lintegral}),
\begin{equation}
\int d\Omega_{y}P_{l_{1}}\left(  \cos\gamma_{\mathbf{ry}}\right)  P_{l_{2}%
}\left(  \cos\theta_{y}\right)  \frac{4\pi}{2l_{2}+1}\delta_{l_{2},l_{3}%
}=P_{l_{1}}\left(  \cos\theta_{r}\right)  \frac{4\pi}{2l_{1}+1}\delta
_{l_{1},l_{2}}\frac{4\pi}{2l_{2}+1}\delta_{l_{2},l_{3}}.\label{intIryz2}%
\end{equation}
Then, performing the integral with respect to $\gamma$ ($\gamma=\theta_{r}$
when $\mathbf{r}^{\prime}=\left(  r^{\prime},0,0\right)  $) in Eq.
(\ref{Iryz}), we have%
\begin{equation}
\frac{4\pi}{2l_{1}+1}\delta_{l_{1},l_{2}}\frac{4\pi}{2l_{2}+1}\delta
_{l_{2},l_{3}}\int_{-1}^{1}d\cos\theta_{r}P_{l}\left(  \cos\theta_{r}\right)
P_{l_{1}}\left(  \cos\theta_{r}\right)  =\frac{32\pi^{2}}{\left(  2l+1\right)
^{3}}\delta_{l,l_{1}}\delta_{l_{1},l_{2}}\delta_{l_{2},l_{3}}.\label{intIryz3}%
\end{equation}
By Eqs. (\ref{intIryz1}), (\ref{intIryz2}), and (\ref{intIryz3}), we have
\begin{equation}
\mathcal{I}_{3}=32\pi^{2}i^{3l}j_{l}\left(  \frac{ry}{i2\left(  t-\tau\right)
}\right)  j_{l}\left(  \frac{yz}{i2\left(  \tau-\tau^{\prime}\right)
}\right)  j_{l}\left(  \frac{zr}{i2\tau^{\prime}}\right)  .\label{intIryz}%
\end{equation}
Substituting Eq. (\ref{intIryz}) into Eq. (\ref{deltaIryz}) gives
\begin{align}
\delta_{l}^{\left(  2\right)  }\left(  k\right)   &  =64\pi^{4}i^{3l}\frac
{1}{2\pi i}\int_{c-i\infty}^{c+i\infty}dt\frac{e^{k^{2}t}}{t}\int_{0}^{t}%
d\tau\int_{0}^{\tau}d\tau^{\prime}\int_{0}^{\infty}y^{2}dyV\left(  y\right)
\nonumber\\
&  \times\int_{0}^{\infty}z^{2}dzV\left(  z\right)  \frac{\exp\left(
-\frac{y^{2}+z^{2}}{4\left(  \tau-\tau^{\prime}\right)  }\right)  }{\left[
4\pi\left(  \tau-\tau^{\prime}\right)  \right]  ^{3/2}}j_{l}\left(  \frac
{yz}{i2\left(  \tau-\tau^{\prime}\right)  }\right)  \\
&  \times\int_{0}^{\infty}r^{2}dr\frac{\exp\left(  -\frac{r^{2}+y^{2}%
}{4\left(  t-\tau\right)  }\right)  }{\left[  4\pi\left(  t-\tau\right)
\right]  ^{3/2}}\frac{\exp\left(  -\frac{z^{2}+r^{2}}{4\tau^{\prime}}\right)
}{\left(  4\pi\tau^{\prime}\right)  ^{3/2}}j_{l}\left(  \frac{ry}{i2\left(
t-\tau\right)  }\right)  j_{l}\left(  \frac{zr}{i2\tau^{\prime}}\right)
.\label{ps2zhongjian}%
\end{align}
To perform the integral with respect to $r$, by using the integral
representation (\ref{j2}) given in \ref{representationj2}, we rewrite%
\begin{equation}
j_{l}\left(  \frac{ry}{i2\left(  t-\tau\right)  }\right)  j_{l}\left(
\frac{zr}{i2\tau^{\prime}}\right)  =\frac{1}{2}\int_{-1}^{1}d\cos\theta
\frac{\sin\sqrt{\left[  \frac{ry}{i2\left(  t-\tau\right)  }\right]
^{2}+\left(  \frac{zr}{i2\tau^{\prime}}\right)  ^{2}-2\frac{ry}{i2\left(
t-\tau\right)  }\frac{zr}{i2\tau^{\prime}}\cos\theta}}{\sqrt{\left[  \frac
{ry}{i2\left(  t-\tau\right)  }\right]  ^{2}+\left(  \frac{zr}{i2\tau^{\prime
}}\right)  ^{2}-2\frac{ry}{i2\left(  t-\tau\right)  }\frac{zr}{i2\tau^{\prime
}}\cos\theta}}P_{l}\left(  \cos\theta\right)  .
\end{equation}
Then, the integral with respect to $r$ can be worked out,%
\begin{align}
&  \int_{0}^{\infty}r^{2}dr\frac{\exp\left(  -\frac{r^{2}+y^{2}}{4\left(
t-\tau\right)  }\right)  }{\left[  4\pi\left(  t-\tau\right)  \right]  ^{3/2}%
}\frac{\exp\left(  -\frac{z^{2}+r^{2}}{4\tau^{\prime}}\right)  }{\left(
4\pi\tau^{\prime}\right)  ^{3/2}}j_{l}\left(  \frac{ry}{i2\left(
t-\tau\right)  }\right)  j_{l}\left(  \frac{zr}{i2\tau^{\prime}}\right)
\nonumber\\
&  =\frac{1}{2}\int_{-1}^{1}d\cos\theta P_{l}\left(  \cos\theta\right)
\int_{0}^{\infty}r^{2}dr\frac{\exp\left(  -\frac{r^{2}+y^{2}}{4\left(
t-\tau\right)  }\right)  }{\left[  4\pi\left(  t-\tau\right)  \right]  ^{3/2}%
}\frac{\exp\left(  -\frac{z^{2}+r^{2}}{4\tau^{\prime}}\right)  }{\left(
4\pi\tau^{\prime}\right)  ^{3/2}}\nonumber\\
&  \times\frac{\sin\sqrt{\left[  \frac{ry}{i2\left(  t-\tau\right)  }\right]
^{2}+\left(  \frac{zr}{i2\tau^{\prime}}\right)  ^{2}-2\frac{ry}{i2\left(
t-\tau\right)  }\frac{zr}{i2\tau^{\prime}}\cos\theta}}{\sqrt{\left[  \frac
{ry}{i2\left(  t-\tau\right)  }\right]  ^{2}+\left(  \frac{zr}{i2\tau^{\prime
}}\right)  ^{2}-2\frac{ry}{i2\left(  t-\tau\right)  }\frac{zr}{i2\tau^{\prime
}}\cos\theta}}\\
&  =\frac{1}{8\pi}\int_{-1}^{1}d\cos\theta P_{l}\left(  \cos\theta\right)
\frac{1}{\left[  4\pi\left(  t-\tau+\tau^{\prime}\right)  \right]  ^{3/2}}%
\exp\left(  -\frac{y^{2}+z^{2}+2yz\cos\theta}{4\left(  t-\tau+\text{$\tau$%
}^{\prime}\right)  }\right)  .\label{integralr}%
\end{align}
Substituting Eq. (\ref{integralr}) into Eq. (\ref{ps2zhongjian}), we have%
\begin{align}
\delta_{l}^{\left(  2\right)  }\left(  k\right)   &  =8\pi^{3}i^{3l}\frac
{1}{2\pi i}\int_{c-i\infty}^{c+i\infty}dt\frac{e^{k^{2}t}}{t}\int_{0}^{t}%
d\tau\int_{0}^{\tau}d\tau^{\prime}\int_{0}^{\infty}y^{2}dyV\left(  y\right)
\nonumber\\
&  \times\int_{0}^{\infty}z^{2}dzV\left(  z\right)  \frac{\exp\left(
-\frac{y^{2}+z^{2}}{4\left(  \tau-\tau^{\prime}\right)  }\right)  }{\left[
4\pi\left(  \tau-\tau^{\prime}\right)  \right]  ^{3/2}}j_{l}\left(  \frac
{yz}{i2\left(  \tau-\tau^{\prime}\right)  }\right)  \\
&  \times\frac{\exp\left(  -\frac{y^{2}+z^{2}}{4\left(  t-\tau+\text{$\tau$%
}^{\prime}\right)  }\right)  }{\left[  4\pi\left(  t-\tau+\tau^{\prime
}\right)  \right]  ^{3/2}}\int_{-1}^{1}d\cos\theta P_{l}\left(  \cos
\theta\right)  \exp\left(  -\frac{yz\cos\theta}{2\left(  t-\tau+\text{$\tau$%
}^{\prime}\right)  }\right)  .\label{deltagamma}%
\end{align}
Using the expansion (\ref{eizcosa}) and the orthogonality of the Legendre
polynomials, we have
\begin{align}
&  \int_{-1}^{1}d\cos\theta P_{l}\left(  \cos\theta\right)  \exp\left(
-\frac{yz\cos\theta}{2\left(  t-\tau+\text{$\tau$}^{\prime}\right)  }\right)
\nonumber\\
&  =\sum_{l^{\prime}=0}\left(  2l^{\prime}+1\right)  i^{l^{\prime}%
}j_{l^{\prime}}\left(  -\frac{yz}{i2\left(  t-\tau+\text{$\tau$}^{\prime
}\right)  }\right)  \int_{-1}^{1}d\cos\theta P_{l}\left(  \cos\theta\right)
P_{l^{\prime}}\left(  \cos\theta\right)  \\
&  =2i^{l}j_{l}\left(  -\frac{yz}{i2\left(  t-\tau+\text{$\tau$}^{\prime
}\right)  }\right)  .\label{integralexp}%
\end{align}
Substituting Eq. (\ref{integralexp}) into Eq. (\ref{deltagamma}) and setting
$T=\tau-\tau^{\prime}$, we have%
\begin{align}
\delta_{l}^{\left(  2\right)  }\left(  k\right)   &  =16\pi^{3}\frac{1}{2\pi
i}\int_{c-i\infty}^{c+i\infty}dt\frac{e^{k^{2}t}}{t}\int_{0}^{\infty}%
y^{2}dyV\left(  y\right)  \int_{0}^{\infty}z^{2}dzV\left(  z\right)
\nonumber\\
&  \times\int_{0}^{t}d\tau\int_{0}^{\tau}dT\frac{\exp\left(  -\frac
{y^{2}+z^{2}}{4T}\right)  }{\left(  4\pi T\right)  ^{3/2}}\frac{\exp\left(
-\frac{y^{2}+z^{2}}{4\left(  t-T\right)  }\right)  }{\left[  4\pi\left(
t-T\right)  \right]  ^{3/2}}j_{l}\left(  \frac{yz}{i2T}\right)  j_{l}\left(
-\frac{yz}{i2\left(  t-T\right)  }\right)  .\label{deltacpt2}%
\end{align}
Exchanging the order of integrals $\int_{0}^{t}d\tau\int_{0}^{\tau
}dT\rightarrow\int_{0}^{t}dT\int_{T}^{t}d\tau$ and resetting $T=\tau^{\prime}$
give%
\begin{align}
\delta_{l}^{\left(  2\right)  }\left(  k\right)   &  =16\pi^{3}\frac{1}{2\pi
i}\int_{c-i\infty}^{c+i\infty}dt\frac{e^{k^{2}t}}{t}\int_{0}^{\infty}%
y^{2}dyV\left(  y\right)  \int_{0}^{\infty}z^{2}dzV\left(  z\right)
\nonumber\\
&  \times\int_{0}^{t}d\tau^{\prime}\int_{\tau^{\prime}}^{t}d\tau\frac
{\exp\left(  -\frac{y^{2}+z^{2}}{4\tau^{\prime}}\right)  }{\left(  4\pi
\tau^{\prime}\right)  ^{3/2}}\frac{\exp\left(  -\frac{y^{2}+z^{2}}{4\left(
t-\tau^{\prime}\right)  }\right)  }{\left[  4\pi\left(  t-\tau^{\prime
}\right)  \right]  ^{3/2}}j_{l}\left(  \frac{yz}{i2\tau^{\prime}}\right)
j_{l}\left(  -\frac{yz}{i2\left(  t-\tau^{\prime}\right)  }\right)  .
\end{align}
Integrating with respect to $\tau$, we have%
\begin{equation}
\delta_{l}^{\left(  2\right)  }\left(  k\right)  =\frac{1}{4}\int_{0}^{\infty
}y^{2}dyV\left(  y\right)  \int_{0}^{\infty}z^{2}dzV\left(  z\right)  \frac
{1}{2\pi i}\int_{c-i\infty}^{c+i\infty}dte^{k^{2}t}\mathcal{I}_{4}%
,\label{deltaLaplace}%
\end{equation}
where
\begin{align}
\mathcal{I}_{4} &  =\frac{1}{t}\int_{0}^{t}d\tau^{\prime}\frac{\exp\left(
-\frac{y^{2}+z^{2}}{4\left(  t-\tau^{\prime}\right)  }\right)  }{\left(
t-\tau^{\prime}\right)  ^{1/2}}\frac{\exp\left(  -\frac{y^{2}+z^{2}}%
{4\tau^{\prime}}\right)  }{\tau^{\prime3/2}}j_{l}\left(  -\frac{yz}{i2\left(
t-\tau^{\prime}\right)  }\right)  j_{l}\left(  \frac{yz}{i2\tau^{\prime}%
}\right)  \nonumber\\
&  =-8\int_{0}^{\infty}kdke^{-k^{2}t}%
\genfrac{\{}{.}{0pt}{0}{k^{2}j_{l}\left(  ky\right)  n_{l}\left(  ky\right)
j_{l}^{2}\left(  kz\right)  ,\text{ \ \ }y>z}{k^{2}j_{l}^{2}\left(  ky\right)
j_{l}\left(  kz\right)  n_{l}\left(  kz\right)  ,\text{ \ \ }y<z}%
,
\end{align}
where $n_{l}\left(  z\right)  $ is the spherical Bessel function of the second
kind. Thus, the inverse Laplace transformation of $\mathcal{I}_{4}$ can be
worked out:%
\begin{equation}
\frac{1}{2\pi i}\int_{c-i\infty}^{c+i\infty}dte^{k^{2}t}\mathcal{I}_{4}=-4%
\genfrac{\{}{.}{0pt}{0}{k^{2}j_{l}\left(  ky\right)  n_{l}\left(  ky\right)
j_{l}^{2}\left(  kz\right)  ,\text{ \ \ }y>z}{k^{2}j_{l}^{2}\left(  ky\right)
j_{l}\left(  kz\right)  n_{l}\left(  kz\right)  ,\text{ \ \ }y<z}%
.\label{Laplace}%
\end{equation}
Substituting Eq. (\ref{Laplace}) into Eq. (\ref{deltaLaplace}), we have%
\begin{align}
\delta_{l}^{\left(  2\right)  }\left(  k\right)   &  =-k^{2}\int_{0}^{\infty
}y^{2}dyj_{l}\left(  ky\right)  n_{l}\left(  ky\right)  V\left(  y\right)
\int_{0}^{y}z^{2}dzj_{l}^{2}\left(  kz\right)  V\left(  z\right)  \nonumber\\
&  -k^{2}\int_{0}^{\infty}y^{2}dyj_{l}^{2}\left(  ky\right)  V\left(
y\right)  \int_{y}^{\infty}z^{2}dzj_{l}\left(  kz\right)  n_{l}\left(
kz\right)  V\left(  z\right)  .\label{detal2twopart}%
\end{align}
By exchanging the order of integrals, $\int_{0}^{\infty}dy\int_{y}^{\infty
}dz\rightarrow\int_{0}^{\infty}dz\int_{0}^{z}dy$, we rewrite Eq.
(\ref{detal2twopart}) as
\begin{align}
\delta_{l}^{\left(  2\right)  }\left(  k\right)   &  =-k^{2}\int_{0}^{\infty
}y^{2}dyj_{l}\left(  ky\right)  n_{l}\left(  ky\right)  V\left(  y\right)
\int_{0}^{y}z^{2}dzj_{l}^{2}\left(  kz\right)  V\left(  z\right)  \nonumber\\
&  -k^{2}\int_{0}^{\infty}z^{2}dzj_{l}\left(  kz\right)  n_{l}\left(
kz\right)  V\left(  z\right)  \int_{0}^{z}y^{2}dyj_{l}^{2}\left(  ky\right)
V\left(  y\right)  .\label{detal2twopart2}%
\end{align}
Obviously, the two parts in Eq. (\ref{detal2twopart2}) are equal. Using
$j_{l}\left(  z\right)  =\sqrt{\pi/\left(  2z\right)  }J_{l+1/2}\left(
z\right)  $ and $n_{l}\left(  z\right)  =\sqrt{\pi/\left(  2z\right)
}Y_{l+1/2}\left(  z\right)  $ gives Eq. (\ref{xiebianps2}).

\section{Comparison with Born approximation \label{bornapp}}

The approach for scattering problems established in the present paper is to
convert a method of calculating heat kernels into a method of solving
scattering problems. As an application, in Sect. \ref{covariantperturbationps}%
, we suggest a method for the scattering phase shift, based on the covariant
perturbation theory of heat kernels.

In scattering theory, there are many approximation methods, such as the Born
approximation, the WKB method, the eikonal approximation, and the variational
method \cite{joachain1975quantum}.

In this section, we compare our method with the Born approximation.

\subsection{Comparison of first-order contribution}

For clarity, we list the result given by the above section in the following.

The first-order phase shift given by covariant perturbation theory given in
Sect. \ref{covariantperturbationps} reads%
\begin{equation}
\delta_{l}^{\left(  1\right)  }\left(  k\right)  _{\text{cpt}}=-\frac{\pi}%
{2}\int_{0}^{\infty}rdrJ_{l+1/2}^{2}\left(  kr\right)  V\left(  r\right)  .
\label{psCP1}%
\end{equation}
For comparison, the first-order phase shift given by the Born approximation
reads \cite{joachain1975quantum}%
\begin{equation}
\delta_{l}^{\left(  1\right)  }\left(  k\right)  _{\text{Born}}=\arctan\left[
-\frac{\pi}{2}\int_{0}^{\infty}rdrV\left(  r\right)  J_{l+1/2}^{2}\left(
kr\right)  \right]  \simeq-\frac{\pi}{2}\int_{0}^{\infty}rdrJ_{l+1/2}%
^{2}\left(  kr\right)  V\left(  r\right)  +\cdots. \label{psB1}%
\end{equation}

Obviously, the leading contributions of these two methods are the same (in the
Born approximation, the first-order contribution is in fact $\arctan\left[
-\left(  \pi/2\right)  \int_{0}^{\infty}rdrV\left(  r\right)  J_{l+1/2}%
^{2}\left(  kr\right)  \right]  $, but the higher contribution can be safely
ignored in the first-order contribution).

\subsection{Comparison of second-order contribution}

The second-order phase shift given by covariant perturbation theory given in
Sect. \ref{covariantperturbationps} [Eqs. (\ref{xiebianps2}) and
(\ref{detal2twopart2})] reads%
\begin{align}
\delta_{l}^{\left(  2\right)  }\left(  k\right)  _{\text{cpt}}  &  =-\frac
{\pi^{2}}{2}\int_{0}^{\infty}rdrJ_{l+1/2}\left(  kr\right)  Y_{l+1/2}\left(
kr\right)  V\left(  r\right)  \int_{0}^{r}r^{\prime}dr^{\prime}J_{l+1/2}%
^{2}\left(  kr^{\prime}\right)  V\left(  r^{\prime}\right) \nonumber\\
&  =-\frac{\pi^{2}}{4}\int_{0}^{\infty}rdrJ_{l+1/2}\left(  kr\right)
Y_{l+1/2}\left(  kr\right)  V\left(  r\right)  \int_{0}^{r}r^{\prime
}dr^{\prime}J_{l+1/2}^{2}\left(  kr^{\prime}\right)  V\left(  r^{\prime
}\right) \nonumber\\
&  -\frac{\pi^{2}}{4}\int_{0}^{\infty}rdrJ_{l+1/2}^{2}\left(  kr\right)
V\left(  r\right)  \int_{r}^{\infty}r^{\prime}dr^{\prime}J_{l+1/2}\left(
kr^{\prime}\right)  Y_{l+1/2}\left(  kr^{\prime}\right)  V\left(  r^{\prime
}\right)  . \label{psCP2}%
\end{align}
The second-order phase shift given by the Born approximation
\cite{joachain1975quantum} reads%
\begin{align}
\delta_{l}^{\left(  2\right)  }\left(  k\right)  _{\text{Born}}  &
=\arctan\left[  -\frac{\pi^{2}}{4}\int_{0}^{\infty}rdrJ_{l+1/2}\left(
kr\right)  Y_{l+1/2}\left(  kr\right)  V\left(  r\right)  \int_{0}%
^{r}r^{\prime}dr^{\prime}J_{l+1/2}^{2}\left(  kr^{\prime}\right)  V\left(
r^{\prime}\right)  \right. \nonumber\\
&  -\left.  \frac{\pi^{2}}{4}\int_{0}^{\infty}rdrJ_{l+1/2}^{2}\left(
kr\right)  V\left(  r\right)  \int_{r}^{\infty}r^{\prime}dr^{\prime}%
J_{l+1/2}\left(  kr^{\prime}\right)  Y_{l+1/2}\left(  kr^{\prime}\right)
V\left(  r^{\prime}\right)  \right] \nonumber\\
&  \simeq-\frac{\pi^{2}}{4}\int_{0}^{\infty}rdrJ_{l+1/2}\left(  kr\right)
Y_{l+1/2}\left(  kr\right)  V\left(  r\right)  \int_{0}^{r}r^{\prime
}dr^{\prime}J_{l+1/2}^{2}\left(  kr^{\prime}\right)  V\left(  r^{\prime
}\right) \nonumber\\
&  -\frac{\pi^{2}}{4}\int_{0}^{\infty}rdrJ_{l+1/2}^{2}\left(  kr\right)
V\left(  r\right)  \int_{r}^{\infty}r^{\prime}dr^{\prime}J_{l+1/2}\left(
kr^{\prime}\right)  Y_{l+1/2}\left(  kr^{\prime}\right)  V\left(  r^{\prime
}\right)  +\cdots. \label{psB2}%
\end{align}

It can be directly seen that the leading contribution of the second-order Born
approximation and the leading contribution of the second-order covariant
perturbation theory are the same.

\subsection{Comparison through an exactly solvable potential£º$V\left(
r\right)  =\alpha/r^{2}$ \label{example}}

In this section, we compare the two methods, the covariant perturbation theory
method, and the Born approximation, through an exactly solvable potential:%
\begin{equation}
V\left(  r\right)  =\frac{\alpha}{r^{2}}. \label{rquar}%
\end{equation}
Using these two approximation methods to calculate an exactly solvable
potential can help us to compare them intuitively.

The phase shift for the potential (\ref{rquar}) can be solved exactly,
\begin{equation}
\delta_{l}=-\frac{\pi}{2}\left[  \sqrt{\left(  l+\frac{1}{2}\right)
^{2}+\alpha}-\left(  l+\frac{1}{2}\right)  \right]  . \label{deltarsquare}%
\end{equation}
In order to compare the methods term by term, we expand the exact result
(\ref{deltarsquare}) as $\delta_{l}=\delta_{l}^{\left(  1\right)  }+\delta
_{l}^{\left(  2\right)  }+\cdots$, where%
\begin{align}
\delta_{l}^{\left(  1\right)  }  &  =-\frac{\pi\alpha}{2(2l+1)}%
,\label{deltarsquareexpand1}\\
\delta_{l}^{\left(  2\right)  }  &  =\frac{\pi\alpha^{2}}{2(2l+1)^{3}}.
\label{deltarsquareexpand2}%
\end{align}

\textit{First order:} The first-order contribution given by covariant
perturbation theory and the Born approximation can be directly obtained by
substituting the potential (\ref{rquar}) into Eqs. (\ref{xiebianps1}) and
(\ref{psB1}), respectively:
\begin{align}
\delta_{l}^{\left(  1\right)  }\left(  k\right)  _{\text{cpt}}  &  =-\frac
{\pi\alpha}{2(2l+1)},\\
\delta_{l}^{\left(  1\right)  }\left(  k\right)  _{\text{Born}}  &
=\arctan\left[  -\frac{\pi\alpha}{2(2l+1)}\right]  \simeq-\frac{\pi\alpha
}{2(2l+1)}-\frac{1}{3}\left[  -\frac{\pi\alpha}{2(2l+1)}\right]  ^{3}.
\end{align}

Comparing with the direct expansion of the exact solution, Eqs.
(\ref{deltarsquareexpand1}) and (\ref{deltarsquareexpand2}), we can see that
both results are good approximations, and the result given by covariant
perturbation theory is better than the result given by the Born approximation.

\textit{Second order:} The second-order contribution given by covariant
perturbation theory and the Born approximation can be directly obtained by
substituting the potential (\ref{rquar}) into Eqs. (\ref{xiebianps2}) and
(\ref{psB2}), respectively:%
\begin{align}
\delta_{l}^{\left(  2\right)  }\left(  k\right)  _{\text{cpt}}  &  =\frac
{\pi\alpha^{2}}{2(2l+1)^{3}},\\
\delta_{l}^{\left(  2\right)  }\left(  k\right)  _{\text{Born}}  &
=\arctan\left[  \frac{\pi\alpha^{2}}{2(2l+1)^{3}}\right]  \simeq\frac
{\pi\alpha^{2}}{2(2l+1)^{3}}-\frac{1}{3}\left[  \frac{\pi\alpha^{2}%
}{2(2l+1)^{3}}\right]  ^{3}.
\end{align}

Comparing with the second-order contribution, Eq. (\ref{deltarsquareexpand2}%
),\ we can see that, like that in the case of first-order contributions, the
result given by covariant perturbation theory is better.

\section{Calculating global heat kernel from phase shift \label{heatkernel}}

The key result of this paper is a relation between partial-wave phase shifts
and heat kernels. Besides solving a scattering problem from a known heat
kernel, obviously, we can also calculate a heat kernel from a known phase
shift. Here, we only give a simple example with the potential $\alpha/r^{2}$.
A systematic discussion of how to calculate heat kernels and other spectral
functions, such as one-loop effective actions, vacuum energies, and spectral
counting functions, from a solved scattering problem will be given elsewhere.

For the potential%
\begin{equation}
V\left(  r\right)  =\frac{\alpha}{r^{2}},
\end{equation}
the exact partial-wave phase shift is given by Eq. (\ref{deltarsquare}),
\begin{equation}
\delta_{l}=-\frac{\pi}{2}\left[  \sqrt{\left(  l+\frac{1}{2}\right)
^{2}+\alpha}-\left(  l+\frac{1}{2}\right)  \right]  .
\end{equation}

By the relation between a global heat kernel and a scattering phase shift
given by Ref. \cite{pang2012relation},
\begin{equation}
K_{l}^{s}\left(  t\right)  =\frac{2}{\pi}t\int_{0}^{\infty}kdk\delta
_{l}\left(  k\right)  e^{-k^{2}t}-\frac{\delta_{l}\left(  0\right)  }{\pi},
\label{pwhkandps}%
\end{equation}
we can calculate the scattering part of the global heat kernel immediately,
\begin{equation}
K_{l}^{s}\left(  t\right)  =-\frac{1}{2}\left[  \sqrt{\alpha+\left(
l+\frac{1}{2}\right)  ^{2}}-\left(  l+\frac{1}{2}\right)  \right]  .
\end{equation}
In this case, the bound part of heat kernel $K_{l}^{b}\left(  t\right)  =0$
and the free part of heat kernel $K_{l}^{f}\left(  t\right)  =R/\sqrt{4\pi
t}-\frac{1}{2}\left(  l+\frac{1}{2}\right)  $, where $R$ is the radius of the
system. The global partial-wave heat kernel then reads
\begin{align}
K_{l}\left(  t\right)   &  =K_{l}^{s}\left(  t\right)  +K_{l}^{b}\left(
t\right)  +K_{l}^{f}\left(  t\right) \nonumber\\
&  =\frac{R}{\sqrt{4\pi t}}-\frac{1}{2}\sqrt{\alpha+\left(  l+\frac{1}%
{2}\right)  ^{2}}. \label{Kltr-2}%
\end{align}

As a comparison, we calculate the partial-wave heat kernel for $V\left(
r\right)  =\alpha/r^{2}$ by another approach.

The partial-wave heat kernel of a free particle, $K_{l}^{f}$, which is the
heat kernel of the radial operator $D^{free}=-\frac{1}{r^{2}}\frac{d}%
{dr}\left(  r^{2}\frac{d}{dr}\right)  +\frac{l\left(  l+1\right)  }{r^{2}}$,
can be calculated directly:
\begin{equation}
K_{l}^{f}\left(  t;r,r^{\prime}\right)  =\frac{1}{2t\sqrt{rr^{\prime}}}%
\exp\left(  -\frac{r^{2}+r^{\prime2}}{4t}\right)  I_{l+1/2}\left(
\frac{rr^{\prime}}{2t}\right)  .
\end{equation}
By setting $\frac{s\left(  s+1\right)  }{r^{2}}=\frac{l\left(  l+1\right)
}{r^{2}}+\frac{\alpha}{r^{2}}$, where $s=\sqrt{\alpha+\left(  l+1/2\right)
^{2}}-1/2$, we can obtain the partial-wave heat kernel of operator
$D_{l}=-\frac{1}{r^{2}}\frac{d}{dr}\left(  r^{2}\frac{d}{dr}\right)
+\frac{l\left(  l+1\right)  }{r^{2}}+\frac{\alpha}{r^{2}}$,
\begin{equation}
K_{l}\left(  t;r,r^{\prime}\right)  =\frac{1}{2t\sqrt{rr^{\prime}}}\exp\left(
-\frac{r^{2}+r^{\prime2}}{4t}\right)  I_{\sqrt{\alpha+\left(  l+1/2\right)
^{2}}}\left(  \frac{rr^{\prime}}{2t}\right)  .
\end{equation}
Taking trace of $K_{l}\left(  t;r,r^{\prime}\right)  $ gives the global
partial-wave heat kernel $K_{l}\left(  t\right)  $:%
\begin{align}
K_{l}\left(  t\right)   &  =\int_{0}^{R}r^{2}drK_{l}\left(  t;r,r\right)
\nonumber\\
&  =\frac{R^{2\left(  1+\eta\right)  }}{\left(  4t\right)  ^{1+\eta}%
\Gamma\left(  2+\eta\right)  }\text{ }_{2}F_{2}\left(  \eta+\frac{1}{2}%
,\eta+1;\eta+2,2\eta+1;-\frac{R^{2}}{t}\right)  ,
\end{align}
where $\eta=\sqrt{(l+1/2)^{2}+\alpha}$, $_{p}F_{q}\left(  a_{1},a_{2}\cdots
a_{p};b_{1},b_{2}\cdots b_{q};z\right)  $ is the generalized hypergeometric
function \cite{olver2010nist}. Expanding $K_{l}\left(  t\right)  $ at
$R\rightarrow\infty$ gives%
\begin{equation}
K_{l}\left(  t\right)  =\frac{R}{\sqrt{4\pi t}}-\frac{1}{2}\eta+\sqrt{t}%
\frac{2\eta^{2}-\frac{1}{2}+ie^{-R^{2}/t-i\pi\eta}}{4\sqrt{\pi}R}+\cdots.
\end{equation}
When $R\rightarrow\infty$, one recovers the heat kernel given by Eq.
(\ref{Kltr-2}).

\section{Conclusions and outlook \label{conclusion}}

In this paper, based on two quantum field theory methods, the heat kernel
method \cite{vassilevich2003heat} and the scattering spectral method
\cite{graham2009spectral}, we suggest an approach for calculating the
scattering phase shift. The method suggested in the present paper is indeed a
series of different methods of calculating scattering phase shifts constructed
from various heat kernel methods.

The key step is to find a relation between partial-wave phase shifts and heat
kernels. This relation allows us to express a partial-wave phase shift by a
heat kernel. Then, each method of the calculation of heat kernels can be
converted to a method of the calculation of phase shifts.

As an application, we provide a method for the calculation of phase shifts
based on the covariant perturbation theory of heat kernels.

Furthermore, as emphasized above, by this approach, we can construct various
methods for scattering problems with the help of various heat kernel methods.
In subsequent works, we shall construct various scattering methods by using
various heat-kernel expansions.

In this paper, as a byproduct, we also provide an off-diagonal heat-kernel
expansion based on the technique developed in the covariant perturbation
theory for diagonal heat kernels, since the heat kernel method for scatterings
established in the present paper is based on the off-diagonal heat kernel
rather than the diagonal heat kernel. It should be emphasized that many
methods for calculating diagonal heat kernels can be directly applied to the
calculation of off-diagonal heat kernels. That is to say, the method for
calculating the diagonal heat kernel often can also be converted to a method
for calculating off-diagonal heat kernels and scattering phase shifts, as we
have done in the present paper. Therefore, we can construct scattering methods
from many methods of diagonal heat kernels, e.g.,
\cite{vassilevich2003heat,fliegner1994higher,nepomechie1985calculating}.

The heat kernel theory is well studied in both mathematics and physics. Here,
as examples, we list some methods on the calculation of heat kernel. In Refs.
\cite{mcavity1991dewitt,branson1999heat,mcavity1991asymptotic}, the authors
calculate the heat-kernel coefficient with different boundary conditions. In
Ref. \cite{van1985explicit}, using the background field method, the author
calculates the fourth and fifth heat-kernel coefficients. In Refs.
\cite{avramidi1989background,avramidi1990covariant,avramidi1991covariant}, the
authors calculate the third coefficient by the covariant technique. In Refs.
\cite{fliegner1994higher,fliegner1998higher}, by a string-inspired worldline
path-integral method, the authors calculate the first seven heat-kernel
coefficients. In Ref. \cite{nepomechie1985calculating}, a direct, nonrecursive
method for the calculation of heat kernels is presented. In Ref.
\cite{van1998index}, the first five heat-kernel coefficients for a general
Laplace-type operator on a compact Riemannian space without boundary by the
index-free notation are given. In Refs.
\cite{barvinsky1987beyond,barvinsky1990covariant,barvinsky1990covariant3,barvinsky1994basis,barvinsky1995one,barvinsky1994asymptotic,mukhanov2007introduction,gusev2009heat,codello2013non,shore2002local}%
, a covariant perturbation theory which yields a uniformly convergent
expansion of heat kernels is established. In Refs.
\cite{gusynin1989new,gusynin1990seeley,gusynin1991heat}, a covariant
pseudo-differential-operator method for calculating heat-kernel expansions in
an arbitrary space dimension is given.

An important application of the method given by this paper is to solve various
spectral functions by a scattering method. The problem of spectral functions
is an important issue in quantum field theory
\cite{dai2009number,dai2010approach,iochum2012spectral}. A subsequent work on
this subject is a systematic discussion of calculating heat kernels, effective
actions, vacuum energies, etc., from a known phase shift. We will show that,
based on scattering methods, we can obtain some new heat-kernel expansions. It
is known that though there are many discussions on the high energy heat-kernel
expansion, the low-energy expansion of heat kernels is relatively difficult to
obtain. While there are some successful low-energy scattering theories, by
using the relation given in this paper we can directly obtain some low-energy
results for heat kernels.

Starting from the result given by the present paper, we can study many
problems. The method presented in this paper can be applied to low-dimensional
scatterings. One-\ and two-dimensional scatterings and their applications have
been thoroughly studied, such as the transport property of low-dimensional
materials \cite{mostafazadeh2014dynamical,mei2003theory,jena2000dislocation}.
We will also consider a systematic application of our method to relativistic
scattering. The relativistic scattering is an important problem, e.g., the
collision of solitons in relativistic scalar field theories
\cite{amin2013scattering} and the Dirac scattering in the problem of the
electron properties of graphene \cite{novikov2007elastic,neto2009electronic}.
We can also apply the method to low-temperature physics. There are many
scattering problems in low-temperature physics, such as the scattering in the
problem of the transition temperature of BEC
\cite{arnold2001bec,kastening2004bose} and the transport property of
spin-polarized fermions at low temperature
\cite{mineev2004transverse,mineev2005theory}.

The application of the method to inverse scattering problems is an important
subject of our subsequent work. The inverse scattering problem has extreme
significance in physics \cite{sabatier2000past,faddeyev1963inverse}. In
practice, for example, the inverse scattering method can be applied to the
problem of BEC \cite{liu2000nonlinear} and the Aharonov--Bohm effect
\cite{nicoleau2000inverse}.

In Ref. \cite{dai2010approach}, we provide a method for solving the spectral
function, such as one-loop effective actions, vacuum energies, and spectral
counting functions in quantum field theory. The key idea is to construct the
equations obeyed by these quantities. We show that, for example, the equation
of the one-loop effective action is a partial integro-differential equation.
By the relation between partial-wave phase shifts and heat kernel, we can also
construct an equation obeyed by phase shifts.

Moreover, in conventional scattering theory, an approximate large-distance
asymptotics is used to seek an explicit result. In Ref.
\cite{liu2014scattering}, we show that such an approximate treatment is not
necessary: without the large-distance asymptotics, one can still rigorously
obtain an explicit result. The result presented in this paper can be directly
applied to the scattering theory without large-distance asymptotics.


\appendix


\section{$\protect\int d\Omega^{\prime}P_{l}\left(  \cos\gamma\right)
P_{l^{\prime}}\left(  \cos\theta^{\prime}\right)  $ \label{legendreintegral}}

In this appendix, we provide an integral formula:
\begin{equation}
\int d\Omega^{\prime}P_{l}\left(  \cos\gamma\right)  P_{l^{\prime}}\left(
\cos\theta^{\prime}\right)  =\frac{4\pi}{2l+1}P_{l}\left(  \cos\theta\right)
\delta_{ll^{\prime}}, \label{Lintegral}%
\end{equation}
where $\gamma$ is the angle between $\mathbf{r=}\left(  r,\theta,\phi\right)
$ and $\mathbf{r}^{\prime}=\left(  r^{\prime},\theta^{\prime},\phi^{\prime
}\right)  $ and $d\Omega^{\prime}=\sin\theta^{\prime}d\theta^{\prime}%
d\phi^{\prime}$.

\begin{proof}
Using the integral formula \cite{flugge1994practical}%
\begin{equation}
\int d\Omega^{\prime}Y_{l0}\left(  \gamma\right)  Y_{l^{\prime}0}\left(
\theta^{\prime}\right)  =\sqrt{\frac{4\pi}{2l+1}}Y_{l0}\left(  \theta\right)
\delta_{ll^{\prime}}%
\end{equation}
and the relation $Y_{l0}\left(  \theta,\phi\right)  =\sqrt{\left(
2l+1\right)  /\left(  4\pi\right)  }P_{l}\left(  \cos\theta\right)  $, we have%
\begin{align}
\int d\Omega^{\prime}Y_{l0}\left(  \gamma\right)  Y_{l^{\prime}0}\left(
\theta^{\prime}\right)   &  =\sqrt{\frac{2l+1}{4\pi}}\sqrt{\frac{2l^{\prime
}+1}{4\pi}}\int d\Omega^{\prime}P_{l}\left(  \cos\gamma\right)  P_{l^{\prime}%
}\left(  \cos\theta^{\prime}\right)  \nonumber\\
&  =P_{l}\left(  \cos\theta\right)  \delta_{ll^{\prime}}.
\end{align}
This proves Eq. (\ref{Lintegral}).
\end{proof}

\section{Integral representations of $j_{l}\left(  u\right)  j_{l}\left(
v\right)  $ and $j_{l}\left(  u\right)  n_{l}\left(  v\right)  $
\label{representationj2}}

In this appendix, we provide two integral representations for the product of
two spherical Bessel functions $j_{l}\left(  u\right)  j_{l}\left(  v\right)
$ and $j_{l}\left(  u\right)  n_{l}\left(  v\right)  $.%
\begin{equation}
j_{l}\left(  u\right)  j_{l}\left(  v\right)  =\frac{1}{2}\int_{-1}^{1}%
d\cos\theta\frac{\sin w}{w}P_{l}\left(  \cos\theta\right)  , \label{j2}%
\end{equation}
where $w=\sqrt{u^{2}+v^{2}-2uv\cos\theta}$ and $l$ is an integer.

\begin{proof}
Using the expansion \cite{watson1944theory}
\begin{equation}
\frac{\sin w}{w}=\sum_{l=0}^{\infty}\left(  2l+1\right)  j_{l}\left(
u\right)  j_{l}\left(  v\right)  P_{l}\left(  \cos\theta\right)
,\label{besselj}%
\end{equation}
where $u=\left\vert \mathbf{u}\right\vert $ and $v=\left\vert \mathbf{v}%
\right\vert $ with $\theta$ the angle between $\mathbf{u}$ and $\mathbf{v}$.
Multiplying both sides of (\ref{besselj}) by $P_{l^{\prime}}\left(  \cos
\theta\right)  $ and integrating from $0$ to $\pi$ give
\begin{equation}
\int_{-1}^{1}d\cos\theta\frac{\sin w}{w}P_{l^{\prime}}\left(  \cos
\theta\right)  =\sum_{l=0}^{\infty}\int_{-1}^{1}d\cos\theta\left(
2l+1\right)  j_{l}\left(  u\right)  j_{l}\left(  v\right)  P_{l}\left(
\cos\theta\right)  P_{l^{\prime}}\left(  \cos\theta\right)  =2j_{l^{\prime}%
}\left(  u\right)  j_{l^{\prime}}\left(  v\right)  .
\end{equation}
Here, the orthogonality, $\int_{-1}^{1}d\cos\theta P_{l}\left(  \cos
\theta\right)  P_{l^{\prime}}\left(  \cos\theta\right)  =2/\left(
2l+1\right)  \delta_{ll^{\prime}}$, is used. This proves Eq. (\ref{j2}).
\end{proof}

\begin{equation}
j_{l}\left(  u\right)  n_{l}\left(  v\right)  =-\frac{1}{2}\int_{-1}^{1}%
d\cos\theta\frac{\cos w}{w}P_{l}\left(  \cos\theta\right)  ,\text{
\ \ }u<v,\label{jn}%
\end{equation}
where $w=\sqrt{u^{2}+v^{2}-2uv\cos\theta}$.

\begin{proof}
Using the expansion \cite{watson1944theory}%
\begin{equation}
\frac{\cos w}{w}=-\sum_{l=0}^{\infty}\left(  2l+1\right)  j_{l}\left(
u\right)  n_{l}\left(  v\right)  P_{l}\left(  \cos\theta\right)  ,\text{
\ \ }u<v,\label{besseljn2}%
\end{equation}
where $u=\left\vert \mathbf{u}\right\vert $ and $v=\left\vert \mathbf{v}%
\right\vert $ with $\theta$ as the angle between $\mathbf{u}$ and $\mathbf{v}%
$. Multiplying both sides of (\ref{besseljn2}) by $P_{l^{\prime}}\left(
\cos\theta\right)  $ and integrating from $0$ to $\pi$ give%
\begin{align}
\int_{-1}^{1}d\cos\theta\frac{\cos w}{w}P_{l^{\prime}}\left(  \cos
\theta\right)   &  =-\sum_{l=0}^{\infty}\left(  2l+1\right)  j_{l}\left(
u\right)  n_{l}\left(  v\right)  \int_{-1}^{1}d\cos\theta P_{l}\left(
\cos\theta\right)  P_{l^{\prime}}\left(  \cos\theta\right)  \nonumber\\
&  =-\sum_{l=0}^{\infty}\left(  2l+1\right)  j_{l}\left(  u\right)
n_{l}\left(  v\right)  \frac{2}{2l+1}\delta_{ll^{\prime}}\nonumber\\
&  =-2j_{l^{\prime}}\left(  u\right)  n_{l^{\prime}}\left(  v\right)  .
\end{align}
This proves Eq. (\ref{jn}).
\end{proof}


\acknowledgments

We are very indebted to Dr G. Zeitrauman for his encouragement. This work is
supported in part by NSF of China under Grant No. 11075115.










\providecommand{\href}[2]{#2}\begingroup\raggedright\endgroup


\end{document}